\documentclass[11pt,twoside,letterpaper]{article}
\usepackage{fullpage}
\usepackage[left=3cm,head=5cm,right=3cm,bottom=3cm]{geometry}
\usepackage{sriram-macros}

\title{Change-Of-Bases Abstractions for Non-Linear Systems.}
\author{Sriram Sankaranarayanan \\
Department of Computer Science \\
University of Colorado, Boulder, CO, USA.\\
\url{srirams@colorado.edu}
}
\date{\today}
\begin{document}




\maketitle

\begin{abstract}
  We present abstraction techniques that transform a given non-linear
  dynamical system into a linear system or an algebraic system
  described by polynomials of bounded degree, such that, invariant
  properties of the resulting abstraction can be used to infer
  invariants for the original system. The abstraction techniques rely
  on a change-of-basis transformation that associates each state
  variable of the abstract system with a function involving the state
  variables of the original system. We present conditions under which
  a given change of basis transformation for a non-linear system can
  define an abstraction. Furthermore, the techniques developed here
  apply to continuous systems defined by Ordinary Differential
  Equations (ODEs), discrete systems defined by transition systems and
  hybrid systems that combine continuous as well as discrete
  subsystems.
  
  The techniques presented here allow us to discover, given a
  non-linear system, if a change of bases transformation involving
  degree-bounded polynomials yielding an algebraic abstraction
  exists. If so, our technique yields the resulting abstract system,
  as well. This approach is further extended to search for a change of
  bases transformation that abstracts a given non-linear system into a
  system of linear differential inclusions.  Our techniques enable the
  use of analysis techniques for linear systems to infer invariants
  for non-linear systems. We present preliminary evidence of the
  practical feasibility of our ideas using a prototype implementation.
\end{abstract}



\section{Introduction}
In this paper, we explore a class of abstractions for non-linear
autonomous systems (continuous, discrete and hybrid systems) using
\emph{Change-of-Bases} (CoB) transformations. CoB transformations are
obtained for a given system by expressing the dynamics of the system
in terms of a new set of variables that relate to the original system
variables through the CoB transformation. Such a transformation is
akin to studying the system under a new set of ``bases''.  We derive
conditions on the transformations such that (a) the CoB
transformations also define an \emph{autonomous system} and (b) the
resulting system abstracts the original system: i.e., all invariants
of the abstract system can be transformed into invariants for the
original system.  Furthermore, we often seek abstract systems through
CoB transformations whose dynamics are of a simpler form, more
amenable to automatic verification techniques.  For instance, it is
possible to use CoB transformations that relate an ODE with non-linear
right-hand sides to an affine ODE, or transformations that reduce the
degree of a system with polynomial right-hand sides. If such
transformations can be found, then safety analysis techniques over the
simpler abstract system can be used to infer safety properties of the
original system.

In this paper, we make two main contributions: (a) we define CoB
transformations for continuous, discrete and hybrid systems and
provide conditions under which a given transformation is valid; (b) we
provide search techniques for finding CoB transformations that result
in a polynomial system whose right-hand sides are degree limited by
some limit $d \geq 1$. Specifically, the case $d=1$ yields an affine
abstraction; and (c) we provide experimental evidence of the
application of our techniques to a variety of ordinary differential
equations (ODEs) and discrete programs.

The results in this paper extend our previously published results that
appeared in HSCC 2011~\cite{Sankaranarayanan/2011/Automatic}. The
contributions of this paper include (a) an extension from linearizing
CoB transformations to degree-bounded polynomial CoB transformations,
(b) extending the theory from purely continuous system to discrete and
hybrid systems, and (c) an improved implementation that can handle
hybrid systems with some evaluation results using this implementation.
On the other hand, our previous work also included an extension of the
theory to differential inequalities and iterative techniques over
cones. These extensions are omitted here in favor of an extended
treatment of the theory of differential equation abstractions for
continuous, discrete and hybrid systems.

\subsection{Motivating Examples}

In this section, we motivate the techniques developed in this paper by
means of a few illustrative examples involving purely continuous ODEs
and purely discrete programs .

Our first example concerns a continuous system defined by a system of
Ordinary Differential Equations (ODEs):
\begin{example}\label{Ex:motivating-example}
Consider a continuous system over $\{x,y\}$:
$ \dot{x} = x y + 2x,\ \  \dot{y}= -\frac{1}{2} y^2 + 7 y + 1$, 
with initial conditions given by the set $ x \in [0,1],\ y \in [0,1]$.
Using the transformation $\alpha: (x,y) \mapsto (w_1,w_2,w_3)$ wherein
$\alpha_1(x,y) = x$, $\alpha_2(x,y) = x y $ and $\alpha_3(x,y) = x
y^2$, we find that the dynamics over $\vec{w}$ can be written as
\[ \dot{w_1} = 2 w_1 + w_2,\  \dot{w_2} = w_1 + 9 w_2 + \frac{1}{2} w_3, \dot{w_3} = 2 w_2 + 16 w_3 \]

Its initial conditions are given by $w_1 \in [0,1],\ w_2 \in
[0,1],\ w_3 \in [0,1]$. We analyze the system using the TimePass tool
as presented in our previous
work~\cite{Sankaranarayanan+Sipma+Manna/06/Fixed} to obtain polyhedral invariants:
\[ \begin{array}{l}
 -w_1 + 2 w_2 \geq -1\ \land\ w_3 \geq 0 \ \land\ w_2 \geq 0\ \land\\
 -16 w_1 + 32 w_2 - w_3 \geq - 17\ \land\ 32 w_2 - w_3 \geq -1\ \land\\
 2 w_1 - 4 w_2 + 17 w_3 \geq -4\ \land\ 286 w_1 - 32 w_2 + w_3 \geq -32 \ \land\ \\
\cdots 
\end{array}\]
Substituting back, we can infer polynomial inequality invariants on
the original system including,
\[\begin{array}{l}
-x + 2 xy \geq -1\ \land\ xy^2 \geq 0\ \land\ -16 x+ 32 xy - xy^2 \geq -17\ \\
x \geq 0\ \land\ 2x - 4 xy + 17 xy^2 \geq -4\ \land\ \cdots 
\end{array}\]

Finally, we integrate the linear system to infer the following
conserved quantity for the underlying non-linear system:
\[ \begin{array}{l}
\left(  \frac{e^{-9 t}}{51}+\frac{1}{102} \left(50+7 \sqrt{51}\right) e^{\left(-9+\sqrt{51}\right) t}+\frac{1}{102} \left(50-7 \sqrt{51}\right) e^{-\left(9+\sqrt{51}\right) t}\right)\ x +  \\
\left(  -\frac{1}{102} e^{-9 t-\left(9+\sqrt{51}\right) t} \left( 
\begin{array}{l} 
7 e^{9 t}-\sqrt{51} e^{9 t}-14 e^{\left(9+\sqrt{51}\right) t}+\\
7 e^{9 t+\left(-9+\sqrt{51}\right) t+\left(9+\sqrt{51}\right) t}+ \\
\sqrt{51} e^{9 t+\left(-9+\sqrt{51}\right) t+\left(9+\sqrt{51}\right) t}
\end{array}\right)
\right)\ xy + \\
\left( \frac{1}{204} e^{-9 t-\left(9+\sqrt{51}\right) t} \left(e^{9 t}-2 e^{\left(9+\sqrt{51}\right) t}+e^{9 t+\left(-9+\sqrt{51}\right) t+\left(9+\sqrt{51}\right) t}\right) \right) xy^2 
\end{array}\]

Finally, if $x(0) \not = 0$, the map $\alpha$ is invertible and
therefore, the ODE above can be integrated.

Note that not every transformation yields a linear abstraction. In
fact, most transformations will not define an abstraction. The
conditions for an abstraction are discussed in
Section~\ref{Section:covAbstraction}. \hfill\halmos

\end{example}

\begin{figure}[t]
\begin{tabular}{cc}
\begin{minipage}{5.5cm}
\begin{lstlisting}
proc computeP(int k)
  int x,y; 
  assert( K > 0);
  x := y := 0;
  while ( y < k ){
       x := x + y * y;
       y := y + 1;
  }
end-function
\end{lstlisting}
\end{minipage} & 
\begin{minipage}{5.5cm}
\begin{lstlisting}
proc computePAbs(int k)
  int x,y,y2; 
  assert( K > 0);
  x := y := y2 := 0;
  while ( y < k ){
       x := x + y2;
       y2 := y2 + 2 * y + 1;
       y := y + 1;
  }
end-function
\end{lstlisting}
\end{minipage} 
\end{tabular}
 
\caption{Program showing a benchmark example proposed by Petter~\cite{Petter/2004/Invarianten} and its abstraction obtained by a change of basis $(x \mapsto x, y \mapsto y, y2 \mapsto y^2)$. }\label{Fig:motivating-example-discrete}
\end{figure}

Next, we motivate our approach on purely discrete programs, showing
how CoB transformations can linearize a discrete program with
non-linear assignments, modeled by a \emph{transition
  system}~\cite{Manna+Pnueli/95/Temporal}. In turn, we show how
invariants of the abstract linearized program can be transferred back.

\begin{example}\label{Ex:motivating-example-discrete}
Figure~\ref{Fig:motivating-example-discrete} shows an example proposed
originally by Petter~\cite{Petter/2004/Invarianten} that considers a
program that sums up all squares from $1$ to $K^2$ for some input
$K \geq 0$. Consider a very simple change of basis transformation
wherein we add a new variable ``y2'' that tracks the value of $y^2$ as
the loop is executed. It is straightforward to write assignments for
``y2'' in terms of itself, $x,y$. Doing so for this example does not
necessitate the tracking of higher degree terms such as $y^3, x^2y^2$
and so on. Finally, the resulting program has affine guards and
assignments, making it suitable for polyhedral abstract
interpretation~\cite{Cousot+Halbwachs/78/Automatic,
Halbwachs+Proy/97/Verification}. The polyhedral analysis yields linear
invariants at the loop head and the function exit in terms of the
variables $x,y,y2$. We may safely substitute $y^2$ in place of $y2$
and obtain invariants over the original program. The non-linear
invariants obtained at the function exit are shown below:
\[ \begin{array}{c}
4  x + 18  y - 7  y^2 \geq 11 \ \land\   4   \leq  2  x + 7  y - 3  y^2 \ \land\   9   \leq   x +12  y - 3  y^2 \ \land\   1   \leq  y \ \land\   \\ 
 3  y - y^2  \leq 2  \ \land\   5  y - y^2  \leq 6 \ \land\    6  y - y^2  \leq 9 \ \land\   k = y 
\end{array}\]

In this example, the change of basis to $y^2$ can, perhaps, be inferred from the
syntax of this program. However, we demonstrate other situations in
this paper, wherein the change of basis cannot be inferred from the
expressions in the program using syntactic means.  

The invariant
\[ 6x = 2 k^3 + 3 k^2 + k \,,\]
discovered by Petter and many other subsequent
works such as the complete approach for P-solvable loops by
Kovacs~\cite{Kovacs/2008/Reasoning} can also be discovered by Karr's
analysis when the term $y^3$ is introduced into the change-of-basis
transformations in addition to $y^2$.\ \hfill\halmos
\end{example}

\subsection{Related Work}

Many different types of \emph{discrete abstractions} have been studied
for hybrid systems~\cite{Alur+Others/2000/Discrete} including
predicate abstraction~\cite{Tiwari/2008/Abstractions} and abstractions
based on invariants~\cite{Oishi+Others/2008/Invariance}.  The use of
counter-example guided iterative abstraction-refinement has also been
investigated in the past (Cf. Alur et
al.~\cite{Alur+Dang+Ivancic/2003/Counter} and Clarke et
al.~\cite{Clarke+Others/2003/Counterexample}, for example). In this
paper, we consider continuous abstractions for continuous systems
specified as ODEs, discrete systems and hybrid systems using a change
of bases transformation. As noted above, not all transformations can
be used for this purpose.  Our abstractions for ODEs bear similarities
to the notion of topological semi-conjugacy between flows of dynamical
systems~\cite{Meiss/2007/Differential}.

Previous work on invariant generation for hybrid system by the author
constructs invariants by assuming a desired template form (ansatz)
with unknown parameters and applying the ``consecution'' conditions
such as \emph{strong consecution} and \emph{constant scale}
consecution~\cite{Sankaranarayanan+Sipma+Manna/2008/Constructing}.
Matringe et al. present generalizations of these conditions using
morphisms~\cite{Matringe+Others/2009/Morphisms}. Therein, they observe
that strong and constant scale consecution conditions correspond to a
linear abstraction of the original non-linear system of a restrictive
form. Specifically, the original system is abstracted by a system of
the form $\frac{dx}{dt} = 0$ for strong consecution, and a system of
the form $\frac{dx}{dt} = \lambda x$ for constant-scale
consecution. This paper builds upon this observation by Matringe et
al. using fixed-point computation techniques to search for a general
linear abstraction that is related to the original system by a change
of basis transformation.  Our work is also related to the technique of
differential invariants proposed by Platzer et
al.~\cite{Platzer+Clarke/09/Differential}. At a high level Platzer et
al. attempt to prove an invariant $p = 0$ for a continuous system
(often a subsystem of a larger hybrid system) using differential
invariant rule wherein the state assertion $\frac{dp}{dt} = 0$ is
established. Likewise, to prove $p \leq 0$, it seeks to establish
$\frac{dp}{dt} \leq 0$. In this paper, we may view the same process
through a CoB transformation $ w \mapsto p(x)$ that allows us to write
the abstract dynamics as $\frac{dw}{dt} = 0$. Going further, we seek
to compute $\vec{w} \mapsto \alpha(\vec{x})$ that maps the dynamics to
an affine or a polynomial system. On the other hand, differential
invariants allow us to reason about Boolean combinations of assertions
and embed into a rich dynamic-logic framework combining discrete and
continuous actions on the state. The work here and its extension to
differential inequalities~\cite{Sankaranarayanan/2011/Automatic} can
be utilized in such a framework.

Fixed point techniques for deriving invariants of differential
equations have been proposed by the author in previous
papers~\cite{Sankaranarayanan+Sipma+Manna/06/Fixed,Sankaranarayanan/2010/Invariant}
These techniques have addressed the derivation of polyhedral
invariants for affine
systems~\cite{Sankaranarayanan+Sipma+Manna/06/Fixed} and algebraic
invariants for systems with polynomial right-hand
sides~\cite{Sankaranarayanan/2010/Invariant}.  In this technique, we
employ the machinery of fixed-points. Our primary goal is not to
derive invariants, per se, but to search for abstractions of
non-linear systems into linear systems. 

\paragraph{Discrete Systems:} There has been a large body of work focused on the use of algebraic
techniques for deriving invariants of programs. Previous work by the
author focuses on deriving polynomial equality invariants for
programs, automatically, by setting up template polynomial invariants
with unknown coefficients and deriving constraints on values of these
coefficients to ensure
invariance~\cite{Sankaranarayanan+Sipma+Manna/2008/Constructing,Sankaranarayanan+Others/2004/Non-Linear}.
Carbonell et al. present loop invariant generation techniques by
solving recurrences and computing polynomial ideas to capture
algebraic properties of the reachable
states~\cite{Carbonell+Kapur/2004/Automatic} and subsequently using
the descending abstract interpretation over ideals with widening over
ideals to ensure termination~\cite{Carbonell+Kapur/2004/Abstract}. The
approach is extended to polyhedral cones generated by polynomial
inequalities to generate polynomial inequality
invariants~\cite{Bagnara+Carbonell+Zaffanella/2005/Generation}.
Another set of related techniques concern the use of linear invariant
generation techniques for polynomial equality invariant generation.
M{\"u}ller-Olm and Seidl explore the use of linear algebraic
techniques, wherein a vector space of matrices are used to summarize
the transformation from the initial state of a program to a given
location. This space is then used to generate polynomial invariants of
the program~\cite{Muller-Olm+Seidl/2004/Precise}. Likewise, the work
of Col{\'o}n explores degree-bounded restrictions to Nullstellensatz
to enable linear algebraic techniques to generate polynomial
invariants~\cite{Colon/04/Approximating}. More recently, the work of
Kovacs uses sophisticated techniques for solving recurrence equations
over so-called P-solvable loops to generate polynomial invariants for
them~\cite{Kovacs/2008/Reasoning}.

Finally, our approach is closely related to
\emph{Carlemann embedding} that can be used to linearize a given
differential equation with polynomial right-hand
sides~\cite{Kowalski+Steeb/1991/Non-Linear}. The standard Carlemann
embedding technique creates an infinite dimensional linear system,
wherein, each dimension corresponds to a monomial or a basis
polynomial. In practice, it is possible to create a linear
approximation with known error bounds by truncating the monomial terms
beyond a degree cutoff.  Our approach for differential equation
abstractions can be \emph{roughly} seen as a search for a ``finite
submatrix'' inside the infinite matrix created by the Carleman
linearization. The rows and columns of this submatrix correspond to
monomials such that the derivative of each monomial in the submatrix
is a linear combination of monomials that belong the submatrix.  Note,
however, that while Carleman embedding is defined using some basis for
polynomials (usually power-products), our approach can derive
transformations that may involve polynomials as opposed to just
power-products.

\paragraph{Organization:} The rest of this paper presents our approach
for Ordinary Differential Equations in
Section~\ref{Section:covAbstraction}. The ideas for discrete systems
are presented in Section~\ref{Section:prgTransformation} by first
presenting the theory for simple loops and then extending it to
arbitrary discrete programs modeled by transition systems. The
extensions to hybrid systems are presented briefly by suitably merging
the techniques for discrete programs with those for ODEs. Finally,
Section~\ref{Section:implementation} presents an evaluation of the
ideas presented using our implementation that combines an automatic
search for CoB transformations with polyhedral invariant generation
for continuous, discrete and hybrid
systems~\cite{Cousot+Halbwachs/78/Automatic,Halbwachs+Proy/97/Verification,Sankaranarayanan+Sipma+Manna/06/Fixed}.

\section{Abstractions for ODEs}\label{Section:covAbstraction}
We first present some preliminary definitions for continuous systems
defined by Ordinary Differential Equations (ODEs).
\subsection{Preliminaries: Continuous Systems}\label{Section:prelims}
Let $\reals$ denote the field of real numbers.  Let $x_1,\ldots,x_n$
denote a set of variables, collectively represented as $\vx$. The set
$\reals[\vx]$ denotes the ring of multivariate polynomials over
$\reals$.

A \emph{power-product} over $\vx$ is of the form $x_1^{r_1}
x_2^{r_2}\cdots x_n^{r_n}$, succinctly written as $\vx^{\vec{r}}$,
wherein each $r_i \in \mathbb{N}$. The \emph{degree} of a monomial
$\vx^{\vec{r}}$ is given by $\sum_{i=1}^n r_i = \vec{1}\cdot \vec{r}$.
A \emph{monomial} is of the form $c \cdot m$ where $c \in \reals$ and
$m$ is a power-product.  A multivariate polynomial $p$ is a sum of
finitely many monomial terms: $p = \sum_{\vec{r} \in \reals^n}
c_{r} \vx^{\vr}$. The degree of a multivariate polynomial $p$
is the maximum over the degrees of all monomial terms $m$ that \emph{occur}
in $p$ with a non-zero coefficient.
 
We assume some basic familiarity with the basics of computational
algebraic geometry~\cite{Cox+Others/1991/Ideals} and elementary linear
algebra~\cite{Halmos/74/Finite}.

\paragraph{Vector Fields:} A \emph{vector field} $F$ over a manifold
$M \subseteq \reals^n$ is a map $F: M \mapsto \reals^n$ from each $\vx
\in M$ to a vector $F(\vx) \in \reals^n$, wherein $F(\vx) \in
T_M(\vx)$, the tangent space of $M$ at $\vx$.

A vector field $F$ is continuous if the map $F$ is continuous. A
\emph{polynomial vector field} $F \in (\reals[\vx])^n$ is
specified by a tuple $ F(\vx)= \tupleof{p_1(\vx),
  p_2(\vx),\ldots,p_n(\vx)}$, wherein $p_1,\ldots,p_n \in
\reals[\vx]$.  

A system of (coupled) ordinary differential equations (ODE) specifies
the evolution of variables $\vx:(x_1,\ldots,x_n) \in M$ over time $t$:
\[ \diff{x_1}{t} = p_1(x_1,\ldots,x_n),\ \cdots,\ \diff{x_n}{t} = p_n(x_1,\ldots,x_n)\,,\] 
The system implicitly defines a vector field
$F(\vx): \tupleof{p_1(\vx),\ldots,p_n(\vx)}$.  We assume that all
vector fields $F$ considered in this paper are (locally) Lipschitz
continuous over the domain $M$.  In general, all polynomial vector
fields are locally Lipschitz continuous, but not necessarily
\emph{globally} Lipschitz continuous over an unbounded domain $X$.
The Lipschitz continuity of the vector field $F$, ensures that given
$\vx=\vx_0$, there exists a time $T > 0$ and a unique time trajectory
$\tau: [0,T) \mapsto \reals^n$ such that $\tau(t) =
\vx_0$~\cite{Meiss/2007/Differential}.

\begin{definition}
For a vector field $F:\ \tupleof{f_1,\ldots,f_m}$, the \emph{Lie derivative}
of a smooth function $f(\vx)$ is given by
\[\lie_F(f) = (\grad f)\cdot F(\vx)= \mathop{\sum}_{i=1}^n \left( \pdiff{f}{x_i} \cdot f_i\right) \]
\end{definition}

Henceforth, wherever the vector field $F$ is clear from the context,
we will drop subscripts and use $\lie(p)$ to denote the Lie derivative
of $p$ w.r.t  $F$.

\begin{definition}
  A continuous system over variables $x_1,\ldots,x_n$ consists of a
  tuple $\scrS: \tupleof{X_{0}, \F, X_I}$ wherein $X_0 \subseteq
  \reals^n$ is the set of initial states, $\F$ is a vector field over
  the domain represented by a manifold $X_I \subseteq \reals^n$.
\end{definition}

Note that in the context of hybrid systems, the set $X_I$ is often
referred to as the \emph{state invariant} or the \emph{domain}
manifold.




\subsection{Change-of-Bases for Continuous Systems}
In this section, we will present change-of-bases (CoB) transformations
of continuous systems and some of their properties. 

Consider a map $\alpha:\reals^k \mapsto \reals^l$. Given a set $S
\subseteq \reals^k$, let $\alpha(S)$ denote the set obtained by
applying $\alpha$ to all the elements of $S$.  Likewise, the inverse
map over sets is $ \alpha^{-1}(T):\ \{ s\ |\ \alpha(s) \in T \}$.  Let
$\scrS: \tupleof{X_0, \F, X_I}$ be a continuous system over variables
$\vx:\ (x_1,\ldots,x_n)$ and $\scrT: \tupleof{Y_0, \G, Y_I}$ be a
continuous system over variables $\vy: (y_1,\ldots,y_m)$.

\begin{definition}\label{Def:semi-conjugacy}
  We say that $\scrT$  \emph{simulates} $\scrS$ iff there
  exists a smooth mapping $\alpha:\ \reals^n \mapsto \reals^m$ such that
  \begin{enumerate}
    \item $Y_0 \supseteq  \alpha(X_0)$ and  $Y_I \supseteq \alpha(X_I)$.
      \item For any trajectory $\tau: [0,T) \mapsto X_I$\ of\ $\scrS$, 
        $\alpha \circ \tau$ is a trajectory of $\scrT$.
  \end{enumerate}
\end{definition}

A simulation relation implies that any time trajectory of $\scrS$ can
be mapped to a trajectory of $\scrT$ through $\alpha$. However,
since $\alpha$ need not be invertible, the converse
need not hold. I.e, $\scrT$ may exhibit time trajectories that are not
mapped onto by any trajectory in $\scrS$.

Let $\scrS$ and $\scrT$ be defined by Lipschitz continuous vector
fields. The following theorem enables us to check given $\scrS$ and
$\scrT$, if $\scrT$ simulates $\scrS$.

\begin{theorem}\label{Theorem:simulation}
$\scrT$ simulates $\scrS$ if the following conditions hold:
\begin{enumerate}
\item $Y_0 \supseteq  \alpha(X_0)$.
\item  $Y_I \supseteq \alpha(X_I)$.
\item 
$ \G(\alpha(\vec{x})) = J_\alpha.\F(\vec{x})$, wherein, 
$J_{\alpha}$ is the Jacobian matrix
\[ J_{\alpha}(x_1,\ldots,x_n) = \left[ \begin{array}{rcl}
\frac{\partial \alpha_1}{\partial x_1} & \cdots &  \frac{\partial \alpha_1}{\partial x_n} \\
\vdots & \ddots & \vdots \\
\frac{\partial \alpha_m}{\partial x_1} & \cdots & \frac{\partial \alpha_m}{\partial x_n} \\
\end{array}\right] \,,\]
and $\alpha(\vx) = ( \alpha_1(\vx),\cdots,\alpha_m(\vx)),\ \alpha_i:
\reals^n \mapsto \reals$.
\end{enumerate}
\end{theorem}
\begin{proof}
  Let $\tau_x$ be a trajectory over $\vx$ for system $\scrS$. Note
  that at any time instant $t \in [0,t)$, $ \frac{d \tau_x}{d t} =
  \F(\tau(t))$.

We wish to show that $\tau_y(t) = \alpha(\tau_x(t))$ is a time
trajectory for the system $\scrT$. Since, $\tau_x(0) \in X_0$, we conclude that $\tau_y(0) =
\alpha(\tau_x(0)) \in Y_0$. Since $\tau_x(t) \in X_I$ for all $t \in
[0,T)$, we have that $\tau_y(t) = \alpha(\tau_x(t))\in Y_I$.
Differentiating $\tau_y$ we get,
\[ \begin{array}{rclclcl}
\frac{d \tau_y}{dt} &=& \frac{d \alpha(\tau_x(t))}{dt} &=& J_{\alpha}\cdot \frac{d \tau_x}{dt} &=& J_{\alpha} \cdot \F(\tau_x(t)) \\
                    &=& \G (\alpha(\tau_x(t)))                     &=& \G(\tau_y(t))\,. \\
\end{array} \]
Therefore $\tau_y = \alpha \circ \tau_x$ conforms to the dynamics of
$\scrT$. By Lipschitz continuity of $\G$, we obtain that $\tau_y$ is
the unique trajectory starting from $\alpha \circ \tau(0)$.
\end{proof}

Theorem~\ref{Theorem:simulation} shows that the condition 
\[ \G(\alpha(\vec{x})) = J_\alpha.\F(\vec{x}) \] relating vector
fields $\F$ and $\G$ suffices to guarantee that time trajectories
(integral curves) of $\F$ are related to those in $\G$ through the map
$\alpha$. In differential geometric terms, this condition can be
stated as $\F$ is $\alpha$-related to
$\G$~\cite{Lee/2003/Introduction}.

 Note that, in general, a
trajectory $\tau_y(t) = \alpha(\tau_x(t))$ may exist for a longer
interval of time than the interval $[0,T)$ over which $\tau_x$ is
assumed to be defined.

\begin{theorem}\label{Thm:weak-preserve-inv}
Let $\scrT$ simulate $\scrS$ through a map $\alpha$.
If $Y \subseteq Y_I$ is a positive invariant set for $\scrT$ then 
$\alpha^{-1}(Y) \cap X_I$ is a positive invariant set for $\scrS$.
\end{theorem}
\begin{proof}
  Assuming otherwise, let $\tau_x$ be a time trajectory that starts
  from inside $ \alpha^{-1}(Y) \cap X_I$ and has a time instant $t$
  such that $\tau_x(t) \not\in \alpha^{-1}(Y) \cap X_I$.  Since we
  defined time trajectories so that $\tau_x(t) \in X_I$, it follows
  that $\tau_x(t) \not\in \alpha^{-1}(Y)$.  As a result,
  $\alpha(\tau_x(t)) \not\in Y$. Therefore, corresponding to $\tau_x$,
  we define a new trajectory $\tau_y= \alpha\circ \tau_x$ which
  violates the positive invariance of $Y$. This leads to a contradiction.
\end{proof}

Let $\varphi[\vy]$ be an assertion representing an invariant of the
system $\scrT$ that simulates $\scrS$ through CoB transformation
$\alpha$.  The assertion $\varphi[ \vy \mapsto \alpha(\vx)]$ obtained
by substituting $\alpha(\vx)$ in place of occurrences of $\vy$ is an
invariant for the original system. In other words, inverting the map
$\alpha$ simply boils down to substituting  $\alpha(\vx)$ in the
invariants of the abstract system.
An application of the Theorem above is illustrated in
Example~\ref{Ex:motivating-example}.
\begin{example}\label{Ex:Volterra-3d-inv}
Consider a mechanical system $\scrS$ expressed in generalized
position coordinates $(q_1,q_2)$ and momenta $(p_1,p_2)$ defined
using the following vector field:
\[F(p_1,p_2,q_1,q_2): \left\langle\begin{array}{l} -2q_1q_2^2,\ -2q_1^2q_2,\ 2p_1,\ 2p_2
\end{array}\right\rangle\]
with the initial conditions: $ (p_1,p_2) \in [-1,1] \times
[-1,1]\ \land\ (q_1,q_2): (2,2)$.  Using the transformation
$\alpha(p_1,p_2,q_1,q_2): p_1^2+p_2^2+q_1^2q_2^2$, we see that $\scrS$
is simulated by a linear system $\scrT$ over $y$, with dynamics given
by $\frac{d y}{dt} = 0,\ y(0) \in [16,18]$.

Incidentally, the form of the system $\scrT$ above indicates that
$\alpha$ is an expression for a conserved quantity (in this case, the
Hamiltonian) of the system. \hfill\halmos
\end{example}

The main goal of this work is to study CoB transformations that
``simplify'' the system's dynamics either (a) casting a non-algebraic
vector field into one defined algebraically or (b) reducing the degree
of a given algebraic vector field by means of an abstraction. A
special case consists of \emph{linearizing CoB transformations} that
map a non-linear system to one defined by affine dynamics.

Recall that a system $\scrT$ is algebraic if it is described by a
polynomial vector field. Furthermore, $\scrT$ is \emph{affine} if it
is described by an affine vector field $\frac{d \vy}{dt} = A \vy +
\vb$ for an $m\times m$ matrix $A$ and an $m\times 1$ vector $\vb$.

\begin{definition}
  Let $\scrS$ be a (non-linear) system.  We say that $\alpha$ is an
  \emph{algebraizing CoB transformation} if it maps $\scrS$ to an
  algebraic system $\scrT$.

  We say that $\alpha$ is a \emph{linearizing CoB transformation} if
  it maps each trajectory of $\scrS$ to that of an affine system
  $\scrT$.
\end{definition}

\begin{example}
Consider the vector field $\F$
\[ \frac{dx}{dt} = x^3 - 2 x^2 + y^2 + xy ,\ \frac{dy}{dt} = 2 x - 3 x^2 + 2 y^3 \,.\]
Let $\alpha:(x,y) \rightarrow (w_1,w_2,w_3,w_4)$ be defined as 
\[\alpha(x,y): (x,y,x^2,y^2) \]

We can verify that using $\alpha$, we note that $\F$ is simulated by
the vector field $\G$:
\[ \begin{array}{ll}
\frac{dw_1}{dt} = w_1 w_3 - 2 w_3 + w_4 + w_1 w_2, & \frac{dw_2}{dt} = 2 w_1 - 3 w_3 + 2 w_2 w_4 \\
\frac{dw_3}{dt} = -4 w_1 w_3 + 2 w_3^2 + 2 w_2 w_3 + 2 w_1 w_4, & \frac{dw_4}{dt} = 4 w_1 w_2 - 6 w_2 w_3 + 4 w_4^2 \\
\end{array}\]
Note that while $\F$ is a cubic vector field over $\reals^2$, 
$\G$ is a quadratic vector field over $\reals^4$. \halmos
\end{example}

Example~\ref{Ex:motivating-example} illustrates a linearizing CoB
transformation.

The above definition of an algebraizing or linearizing CoB seems
useful, in practice, only if $\alpha$ and $\scrT$ are already
known. We may then use known techniques for reasoning over algebraic
systems or affine systems for safely bounding the reachable set of an
affine system, given some initial conditions, and transform the result
back through substitution to obtain a bound on the reachable set for
$\scrS$.

We now present a technique that searches for a map $\alpha$ to obtain
an algebraic system $\T$ that simulates a given system $\scrS$ through
$\alpha$ such that the vector field describing $\T$ is degree bounded
by a given  degree limit $d > 0$. In particular, if the degree limit
$d$ is set to $1$, then the resulting transformation $\alpha$ is
linearizing.

We ignore the initial condition and invariant, for the time being, and
simply focus on obtaining the dynamics of $\T$. In other words, we
will search for a map ${\alpha}:\ (\alpha_1,\ldots,\alpha_m)$ that
maps $\reals^n$ into $\reals^m$ so that
\[ J_{\alpha}(\vx) \cdot \F(\vx) = G(\alpha(\vx)) \,.\] 

Having found such a map, we may find appropriate over-approximate
initial and invariance conditions for the simulating system $\scrT$,
so that Definition~\ref{Def:semi-conjugacy} holds. Specifically, we
are interested in finding transformations $\alpha$ that ensure that
(a) $G$ is a polynomial vector field and (b) the degrees of
polynomials describing $G$ are degree bounded by the degree limit $d >
0$.

\subsection{Multilinear Abstractions through Dimension Copying}

We first show that any polynomial system of ODEs can be abstracted by
a \emph{multilinear} system. However, doing so may require $\alpha$ to
have many repeated components wherein $\alpha_i(\vx) = \alpha_j(\vx)$
for $i \not= j$.

\begin{definition}
A polynomial $p$ is defined to be \emph{multilinear} if and only if
each power-product in $p$ is of the form $x_1^{r_1} x_2^{r_2} \cdots x_n^{r_n}$
wherein each $r_i = 0\ \mbox{or}\ 1$. 
\end{definition}

\begin{example}
  As an example, the polynomial $p = 2 x_1 x_2 x_3 + x_1 x_3 + 4 x_1 -
  2 x_2 - 1$ is multilinear. On the other hand, the polynomial $q =
  2x_2^2 + x_1 + x_3$ is not, owing to the $x_2^2$ power product.
\end{example}

We first observe that any polynomial ODE may be equivalently written
by means of a multilinear system using a suitably defined $\alpha$.

\begin{theorem}\label{Theorem:transformation-multi-linear}
  Let $\F$ be a polynomial vector field over $\vx \in \reals^n$. There
  is a transformation $\alpha:\reals^n \rightarrow \reals^m$, that
  maps $\F$ to a multilinear system $\G$.
\end{theorem}
\begin{proof}
  Let us write $\F(\vx): (p_1,\ldots,p_n)$ for multivariate polynomials
  $p_1,\ldots,p_n$.  We will assume that the vector field $\F$ is not
  already multi-linear. Therefore, some $p_j$ has a power product that
  is divisible $x_k^r$ for some $r \geq 2$. The idea is to use $r$
  different functions $\alpha_{k,1} = \alpha_{k,2} = \cdots =
  \alpha_{k,r} = x_k$ so that in the transformed system the term
  $x_k^r$ appears as a multilinear product $y_{k,1} y_{k,2} \cdots
  y_{k,r}$.

  In the worst case, the transformation $\alpha$ involves $n \times K$
  components, wherein \[
K=
  \max(\mathsf{degree}(p_1),\ldots,\mathsf{degree}(p_n))\,.\]  Each
  component $\alpha_{i,k}: x_i$ is simply a ``copy'' of the variable
  $x_i$ that ensures multilinearity of the transformed system.
\end{proof}

\begin{example}
Consider the one dimensional system defined by 
\[ \frac{dx}{dt} = 2 x^5 + 3 x^2 + x -5 \,.\]
We use the transformation $\alpha: \reals \rightarrow \reals^5$ wherein
$\alpha_1(x) = \alpha_2 (x) = \cdots = \alpha_5(x) = x$.  Using this transformation, we derive an abstract system defined by the ODE
\[ \frac{d y_j}{dt} = 2 y_1 y_2 y_3 y_4 y_5  + 3 y_1 y_2 + y_1 - 5 \,,\ j = 1,2,\ldots,5. \]
\halmos
\end{example}
Even though there are efficient algorithms for analyzing multi-linear
systems~\cite{Berman+Halasz+Kumar/2007/MARCO}, the transformation in
Theorem~\ref{Theorem:transformation-multi-linear} faces two potential
problems: (a) the dimensionality of the transformed system $\scrT$ can
be as large as the dimensionality of the original system times the
maximum degree of the polynomials in the RHS of the vector field, and
(b) ignoring the implicit equality relationships between the various
dimensions results in a very coarse abstraction while taking them into
account simply gives us the original system back (albeit in a
different form).

\subsection{Independent Transformations}
The rest of this paper, will focus on \emph{independent
  transformations} $\alpha: (\alpha_1,\ldots,\alpha_N)$ wherein each
$\alpha_i$ cannot be written as a linear combination of the remaining
$\alpha_j$s for $ j \not= i$. Assuming independence automatically rules
out the constructions used in
Theorem~\ref{Theorem:transformation-multi-linear}.

In general, computing independent transformations $\alpha$ for any
given ODE is a hard problem. In this paper, we will focus on solutions
that involve searching for an appropriate map $\alpha$, wherein
$\alpha$ is specified to be the linear combination of some fixed,
finite set of basis functions $g_1,\ldots,g_N$. The initial basis is
assumed to be given to our algorithm by the user. Starting from this
initial basis of functions, our algorithm searches for transformations
$\alpha$ whose components can be written as linear combinations
$\sum_{i=1}^N \lambda_j g_j$.

The basis functions could be specified implicitly as the set of all
power products over $\vx$ of degree up to some limit $K > 0$ or the set
of all power products involving the variables $x_i$ and various
non-algebraic functions $\sin(z), \cos(z)$ and $e^{z}$ applied to
these power products.  Having chosen a basis $B=\{g_1,\ldots,g_N\}$
for $\alpha$, we will cast the search for the map $\alpha$ as a vector
space iteration.

Let $\alpha(\vx): ( \alpha_1(\vx),\ldots,\alpha_m(\vx))$ be a
smooth mapping $\alpha:\reals^n \mapsto \reals^m$, wherein each
$\alpha_i: \reals^n \mapsto \reals$.  Recall that
$\lie_F(\alpha_i(\vx))= (\grad \alpha_i) \cdot \F(\vx)$ denotes the Lie
derivative of the function $\alpha_i(\vx)$ w.r.t vector field $\F$.

\begin{lemma}\label{Lemma:useful-1}
$J_{\alpha} \cdot \F(\vx) = \left(\begin{array}{c}
 \lie_F(\alpha_1(\vx)) \\ 
 \lie_F(\alpha_2(\vx)) \\ 
 \vdots \\
 \lie_F(\alpha_m(\vx)) \\
\end{array}\right) $.
\end{lemma}
\begin{proof}
Recall the definition of the Jacobian matrix $J_{\alpha}$:
\[ J_{\alpha}(x_1,\ldots,x_n) = \left[ \begin{array}{rcl}
\frac{\partial y_1}{\partial x_1} & \cdots &  \frac{\partial y_1}{\partial x_n} \\
\vdots & \ddots & \vdots \\
\frac{\partial y_m}{\partial x_1} & \cdots & \frac{\partial y_m}{\partial x_n} \\
\end{array}\right] =  \left[ \begin{array}{c}
\grad \alpha_1\\
\vdots \\
\grad \alpha_m \\
\end{array}\right]\,.\]
Therefore, 
$ J_\alpha.\F = \left( \begin{array}{c}
(\grad \alpha_1)\cdot (\F) \\
(\grad \alpha_2)\cdot (\F) \\
\vdots \\
(\grad \alpha_m)\cdot (\F) \\
\end{array}\right) = \left(\begin{array}{c}
 \lie_F(\alpha_1(\vx)) \\ 
 \lie_F(\alpha_2(\vx)) \\ 
 \vdots \\
 \lie_F(\alpha_m(\vx)) \\
\end{array}\right)$.
\end{proof}

\paragraph{Note:} For the rest of this section, we will fix a vector field $\F$
belonging to a system $\scrS$ as the original system for which we seek
an abstraction. We will simply write $\lie(g)$ to denote the
Lie-derivative of a given function $g$ in place of $\lie_F(g)$.

\subsection{Vector Space Closure}
We first define the vector spaces that will be used in our search.
\begin{definition}
  Let $B = \{g_1,\ldots,g_k\}$ be some finite set of functions wherein
  $g_i: \reals^n \rightarrow \reals^m$ for some fixed $n,m > 0$. The
  \emph{vector space} spanned by $G$ denoted $\vspan(B)$ consists of
  all functions that are linear combinations of $g_i$:
  \[ \vspan(B) = \left\{ \sum_{i=1}^k \lambda_i g_i \ |\ \lambda_i \in
    \reals \right\} \,.\] 

  We assume, without loss of generality, that the elements in $B$ are
  linearly independent.  I.e., no $g_i \in B$ can be written as a
  linear combination of the remaining $g_j \in B$, for $j \not= i$.
\end{definition}

Let $\mathbf{1}$ represent the constant function $\mathbf{1}(\vx) = \vec{1} \in
\reals^m$.  Given a vector space $V = \vspan(B)$, we define the space
of power products of $V$ up to a degree limit $d \geq 1$ as
\[ \pp{V}{d} = \vspan\left(\left\{ g_{i_1} \times g_{i_2} \times \cdots \times g_{i_d}\ |\ g_{i_1},\ldots,g_{i_d} \in B \cup \{ \mathbf{1} \} \right\}\right) \,. \]
In particular, note that $\pp{V}{1} = \vspan( V \cup \{ \mathbf{1} \} )$.

\begin{example}
  Let $B = \{ x, \sin(y) \}$ be our basis set. The vector space
  $V:\ \vspan(B)$ is given by  $\{ a_1 x + a_2 \sin(y)\ |\ a_1,a_2 \in \reals \}$.
  The space $\pp{V}{2}$ is the set
\[ \left\{ a_0 + a_1 x + a_2 \sin(y) + a_3 x \sin(y) + a_4 x^2 + a_5 \sin^2(y)\ |\ a_0,\ldots,a_5 \in \reals \right\} \,.\]
This space is generated by the functions $ \mathbf{1}, x, \sin(y),
x\sin(y), x^2, \sin^2(y) $.  It consists of all polynomials of degree
at most $2$ formed by the functions $x$, $\sin(y)$. The purpose of
adding the function $\mathbf{1}$ is to enable terms of degree $1$ and
$0$ to be considered. \hfill \halmos
\end{example}

Roughly, the main idea behind our approach is to find a vector space
$U$ that satisfies the following closure property:
\[ (\forall\ f \in U)\ \lie(f) \in \pp{U}{d} \,.\] 

In other words, we will search for a vector space $U$, such that
taking the Lie derivative of any element of $U$ yields an element in
$\pp{U}{d}$. Such a vector space $U$ will be called
$d-\mbox{closed}$. Let $U = \vspan\left(\left\{
h_1,\ldots,h_m\right\}\right)$ be a $d-\mbox{closed}$ vector space. We
will prove that $\alpha:\ (h_1,\ldots,h_m)$ maps the original system
$\scrS$ to an algebraic system $\scrT$ with a vector field of degree
at most $d$.

\begin{definition}~\label{Def:d-closed-vector-space}
A vector space $V$ is said to be $d-\mbox{closed}$ under the application of Lie derivatives iff
$ (\forall\ f \in V)\ \lie(f) \in \pp{V}{d}$.
\end{definition}

In order to check whether a given space $V=\vspan(B)$ is
$d-\mbox{closed}$, it suffices to verify the property in
Definition~\ref{Def:d-closed-vector-space} for the elements in $B$.
\begin{lemma}
  A vector space $U = \vspan\left( \left\{ h_1,\ldots,h_m\right\}\right)$ be
    $d-\mbox{closed}$ under Lie derivatives if and only if $\lie(h_i) \in
    \pp{U}{d}$ for $i \in \{1,\ldots,m\}$.
\end{lemma}
\begin{proof}
  If $U$ is $d-\mbox{closed}$ under Lie derivatives then by definition, the
  Lie derivatives of its basis elements $h_i$ should lie in
  $\pp{U}{d}$. We will prove the reverse direction. Let $U$ be such
  that for each basis element $h_i$, we have $\lie(h_i) \in
  \pp{U}{d}$. Any element of $U$ can be written as $f = \sum_{j=1}^k
  a_j h_j$ for $a_j \in \reals$. We have $\lie(f) = \sum_{j=1}^k a_j
  \lie(h_j)$. Since each $\lie(h_j) \in \pp{U}{d}$, we have that
  $\lie(f) \in \pp{U}{d}$. This completes the proof.
\end{proof}

Next, we relate $d-\mbox{closed}$ vector spaces to algebraizing CoB
  transformations.  Let $B=\left\{ h_1,\ldots,h_m\right\}$ and $U
  = \vspan\left( B\right)$ be a $d-\mbox{closed}$ vector space. Let
  $\alpha$ be the map from $\reals^n \rightarrow \reals^m$ defined as
  $\alpha: (h_1,\ldots,h_m)$.
\begin{theorem}\label{Thm:d-closed-thm}
  The map $\alpha$ formed by the basis elements of a $d-$closed vector
  field is an algebraizing transformation from the original system
  $\scrS$ to a system $\scrT$ defined by a polynomial vector field of
  degree at most $d$.
\end{theorem}
\begin{proof}
Since $U$ is $d-\mbox{closed}$, we note that for each $h_i$ in the basis of $U$, we have $ \lie(h_i) \in \pp{U}{d}$.
In other words, we may write $\lie(h_i)$ as a linear combination of power products as shown below:
\begin{equation}\label{Eq:1}
 \lie(h_i):\  \sum_{j=1}^K a_{ij} h_{i,j,1}\times h_{i,j,2} \times \cdots \times h_{i,j,d}\,,\ \mbox{wherein}\ h_{i,j,k} \in B \cup \{ \mathbf{1}\} 
\end{equation}
We define the system $\scrT$ over variables $y_1,\ldots,y_m$.  We will
use variable $y_i$ to correspond to $h_i(\vx)$. The dynamics are obtained as 
\[ \frac{dy_i}{dt} = \sum_{j=1}^K a_{ij} y_{i_1}\times y_{i_2} \times
\cdots \times y_{i_k} \,,\] 
by substituting the variable $y_j$
wherever the function $h_j$ occurs in Equation~\eqref{Eq:1}.  Let $G$
be the resulting vector field on $\vy$. It is easy to see that (a) $G$
is a polynomial vector field and (b) of degree at most $d$.

From Lemma~\ref{Lemma:useful-1}, we note that $J_{\alpha} \F(\vx) =
(\lie(h_1),\ldots, \lie(h_m))$.  We verify that
$(\lie(h_1),\ldots,\lie(h_m)) = G(h_1(\vx),\ldots,h_m(\vx))$. This is
directly evident from the construction of $G$ from
Equation~\eqref{Eq:1}. Thus, the key condition (3) of
Theorem~\ref{Theorem:simulation} is seen to hold. By finding the right
sets $Y_0, Y_I$ given $\alpha$, we take care of the remaining
conditions as well.
\end{proof}

\paragraph{Note:} The trivial space $V = \vspan(\{0\})$ consisting  of the constant function that maps all inputs to $\vec{0}$ is always $d-$closed. This space yields $\alpha: ( 0 )$ that maps all states $\vx$ to the zero vector. As such,
 the map $\alpha$ is not very useful in practice for inferring
 invariants.

\begin{example}\label{Ex:motivating-example-abstraction}
Consider the ODE from Example~\ref{Ex:motivating-example} recalled
below:
\[ \begin{array}{rcl}
\frac{dx}{dt} &=& x y + 2x \\
\frac{dy}{dt} &=& -\frac{1}{2} y^2 + 7 y + 1 \\
\end{array}\]
We claim that the vector space $V$ generated by the set
of functions $ \{  x, xy, xy^2 \} $
is $1-$closed.
To verify, we compute the Lie derivative of a function of the form
$ c_1 x  + c_2 xy + c_3 xy^2$
to obtain
\[ c_1 (xy + 2x) + c_2 (\frac{1}{2} xy^2 + 9 x y + x) + c_3 ( 16 xy^2 + 2 xy) \]
 which is seen to belong to $\pp{V}{1}$. As a result, we obtain the
 CoB abstraction $ \alpha(x,y): ( x, xy, xy^2)$ that maps the vector
 field to an affine vector field (polynomial of degree $1$).

 The abstract system over $(w_1,w_2, w_3) \in \reals^3$ has dynamics
 given by
\[ \begin{array}{rcl}
\frac{dw_1}{dt} &=& 2 w_1 + w_2 \\
\frac{dw_2}{dt} &=& \frac{1}{2} w_3 + 9 w_2 + w_1 \\
\frac{dw_3}{dt} &=& 16 w_3 + 2 w_2 \\
\end{array}\]

The mapping between original and abstract system is given by 
\[ w_1\ \mapsto\ x,\ w_2\ \mapsto\ xy,\ w_3\ \mapsto\ xy^2 \,.\]

\hfill\halmos
\end{example}

\subsection{Finding Closed Vector Spaces}
We will now describe a search technique for finding a map $\alpha$
and the associated abstraction $\scrT$, such that the dynamics of
$\scrT$ are described by polynomials with degree bound $d$. If $d=1$,
the dynamics of $\scrT$ are affine. The inputs to our search procedure
are
\begin{enumerate}
\item The original system $\scrS$ described by a vector field $\F$,
\item The degree limit $d$ for the desired vector field $\scrT$, and
\item An initial basis $B_0 = \{ h_1,\ldots,h_N\}$ of continuous and
  differentiable functions. We may regard the linear combination \[
  c_1 h_1(\vx) + c_2 h_2(\vx) + \ldots + c_N h_N(\vx) \,,\] as
  an \emph{ansatz} or a template for each component $\alpha_j$ of the
  map $\alpha:(\alpha_1,\ldots,\alpha_m)$, that we are searching
  for. However, we do not fix the number of components $m$ of the
  transformation $\alpha$, \emph{apriori}, or guarantee that a
  non-trivial $\alpha$ (with $m > 0$) can be found.
\end{enumerate}

The initial basis $B_0$ is often specified as consisting of all power
products of the variables in $\vx$ with a given degree limit $M$. This
limit $M$ is chosen independent of the limit $d$ for the desired
abstraction $\scrT$.

Our overall approach is to start with the initial vector space $V_0:
\vspan(B_0)$ and iteratively refine $V_0$ to construct a sequence of
vector spaces 
\[ V_0 \supseteq V_1 \supseteq V_2 \cdots \supseteq V_k = V_{k+1}=
V^*\] wherein, (1) $V_{j+1} \subseteq V_j$, for $j \in [1,k-1]$, and
(2) $V_{k} = V_{k+1}$. The iterative scheme is designed to guarantee
that the converged result $V^*$ is $d-$ closed. If $V^*$ has a
non-zero basis, then the basis elements of $V^*$ form the components
of the map $\alpha$ and the abstraction $\scrT$ whose dynamics have
the desired form.

The main step of iteration is to derive $V_{i+1}$ from $V_i$. This is performed as follows:
\begin{equation}~\label{Eq:refineV}
 V_{i+1} = \{ g \in V_i\ |\ \lie(g) \in \pp{V_i}{d} \} \,.
\end{equation}
In other words, $V_{i+1}$ retains those functions $g \in V_i$ whose
Lie derivatives also lie inside $\pp{V_i}{d}$.

\begin{lemma}
  (1) $V_{i+1}$ is a sub-space of $V_i$.  (2) $V_{i}$ is $d-$closed
  iff $V_{i} = V_{i+1}$.
\end{lemma}
\begin{proof}
 We prove the two parts (1) and (2) as follows.

(1) Since by Eq.~\eqref{Eq:refineV}, $V_{i+1} \subseteq V_i$, it
suffices to show that $V_{i+1}$ is a vector space. Let
$g_1,\ldots,g_k \in V_{i+1}$.  We have that $g_1,\ldots,g_k \in
V_i$. Furthermore, since $V_i$ is a vector space, any linear combination
$g:\ \sum_{j=1}^k \lambda_j g_j \in V_i$. The lie derivative $\lie(g)$
can be written as $\sum_{j=1}^k \lambda_j \lie(g_j)$.  Since
$\lie(g_j) \in \pp{V_i}{d}$, we have $\lie(g)
= \sum_{j=1}^k \lambda_j \lie(g_j) \in \pp{V_i}{d}$. Therefore, by definition
$g \in V_{i+1}$ as well. The linear combination of any finite subset
of elements from $V_{i+1}$ also belongs to $V_{i+1}$, proving that it
is a sub-space of $V_i$.

(2) If $V_{i} = V_{i+1}$, it is easy to check that $V_i$ satisfies the
definition of being $d-$ closed. For the other direction, let us
assume that $V_{i}$ is $d-$closed. Then for each $g \in V_i$, we have
$\lie(g) \in \pp{V_i}{d}$. Thus $g \in V_{i+1}$. This proves that
$V_{i+1} \supseteq V_i$. Combining with the fact that
$V_{i+1} \subseteq V_{i}$, we obtain equality.
\end{proof}

We now focus on calculating $V_{i+1}$ from $V_i$. Let
$V_i:\ \vspan(B_i)$ for a finite set $B_i$. Any element of $V_i$ can
be represented as $\sum_{h_j \in B_i} c_j h_j$ for some multipliers
$c_j$. The Lie derivative is expressed as $\sum_{h_j \in B_i}
c_j \lie(h_j)$. The procedure for calculating $V_{i+1}$ reduces to
finding the set of multipliers $(c_1,\ldots,c_M)$ where $M = | B_i|$
such that $\sum_{h_j \in B_i} c_j \lie(h_j) \in \pp{V_i}{d}$.

The key challenge lies in comparing two elements of the form $\sum_j
c_j \lie(h_j)$ and $ \sum_k d_k g_k$, for unknowns $c_j$ and $d_k$,
where $h_j \in B_i$ and $g_k \in \pp{V_i}{d}$. If both the functions
are polynomials over $\vx$, the comparison is performed by equating
the coefficients of corresponding monomials. This is illustrated using
the example below:
\begin{example}
  Consider once again the ODE from Example~\ref{Ex:motivating-example}
  and ~\ref{Ex:motivating-example-abstraction}. We seek to find an
  affine system $\scrT$ that abstracts this system. Let us consider
  the space $V_0$ generated by the basis $B_0: \{ x,y,xy,x^2,y^2 \}$
  of all degree $2$ monomials. Any element in $V_0$ can be written 
  as 
  \[ p(c_1,\ldots,c_5):\ c_1 x + c_2 y + c_3 xy + c_4 x^2 + c_5 y^2 \,.\]
  Its Lie derivative is given by 
\[\begin{array}{l}
 c_1 (xy + 2x) + c_2 ( -\frac{1}{2} y^2 + 7 y + 1) + c_3 x ( -\frac{1}{2} y^2 + 7 y + 1) \\
 + c_3 y ( xy + 2x) + c_4 (2x)(xy + 2x) + c_5 (2y) ( -\frac{1}{2} y^2 + 7 y + 1 ) \end{array}\]
This can be simplified as
\[ p'(c_1,\ldots,c_5):\ \left[ \begin{array}{l} 
c_2 + (2 c_1 + c_3) x + (7 c_2 + 2 c_5) y + (c_1 + 9 c_3 ) xy + 4 c_4 x^2 +\\
 (14 c_5 - \frac{1}{2}c_2) y^2 +  \frac{1}{2} c_3  x y^2 + 2 c_4 x^2 y - c_5 y^3 
\end{array}\right]
\,.\]
We require the Lie derivative to belong to $\pp{V}{1} = \vspan(B_0 \cup \{1\})$.
This yields the constraints:
\[ (\exists d_0, d_1, \ldots , d_5)\ (\forall\ x,y)\ 
 d_0 + d_1 x + d_2 y + d_3 xy + d_4 x^2 + d_5 y^2 = p'(c_1,\ldots,c_5) \,.\]

We use the lemma that two polynomials are identical iff their coefficients
on corresponding power-products are. This yields the following system of linear
equations:
\[ \begin{array}{l}
c_2 = d_0,\ 2 c_1 + c_3 = d_1,\ 7 c_2 + 2 c_5 = d_2,\ c_1 + 9 c_3 = d_3,\\
 4 c_4 = d_4, 14 c_5 - \frac{1}{2} c_2 = d_5,\ c_3 = 0,\ 2c_4 = 0,\  c_5 = 0 \\
\end{array} \]

Eliminating $d_0,\ldots,d_5$, we obtain the constraints $c_3 = c_4 =
c_5 = 0$.  The new basis $B_1$ is $\{ x,y \}$. \ \hfill \halmos
\end{example}

On the other hand, if the basis $B_i$ involves non-polynomials
(trigonometric or exponential functions), then encoding equality by
matching up coefficients of syntactically identical terms is
\emph{incomplete}: I.e, not all solutions can be found by equating
coefficients of matching terms.  In general, deciding if two
expressions involving trigonometric functions is identically zero is
undecidable~\footnote{ This follows from Richardson's
theorem~\cite{Petkovsek+Wilf+Zeilberger/1996/AB}.}. In practice, we
may continue to handle trigonometric functions using the same
syntactic matching technique that is complete for polynomials.  If a
$d-$closed basis is discovered this way, then it may be used to derive
a valid abstraction. On the other hand, the process may be unable to
find a vector space starting from the initial set of functions even if
one such exists.

\begin{example}
Consider a simple example with the ODE
\[ \frac{dx}{dt} = \sin(x+y),\ \ \ \frac{dy}{dt} = x + y  \,.\]
Consider the space $V$ spanned by the basis 
\[ B = \{ x,y,\sin(x), \sin(y), \cos(x), \cos(y) \}\,.\] 
Our goal is
to check if $V$ is $3-$closed.  Any element of $V$ can be written as
\[ c_1 x + c_2 y + c_3 \sin(x) + c_4 \sin(y) + c_5 \cos(x) + c_6 \cos(y) \,.\]
Its Lie derivative can be written as
\[ \begin{array}{c}
c_1 \sin(x+y) + c_2 (x+y) + c_3 \cos(x) \sin(x+y) + c_4 \cos(y) (x+y) \\
- c_5 \sin(x)\sin(x+y) - c_6 \sin(y) (x+y) 
\end{array}\,.\]
Our goal is to check if the Lie derivative belongs to 
$\pp{V}{3}$. We note that a syntactic check for membership yields the
constraints $c_1 = c_3 =  c_5 = 0$.  On the other hand, substituting the 
trigonometric identity 
\[ \sin(x+y) \equiv \sin x \cos y + \sin y \cos x \,,\]
we may indeed verify that the Lie derivative of any element
of $V$ belongs to $\pp{V}{3}$. This yields a degree $3$ algebraization
given by $\alpha(x,y): (x,y,\sin(x),\sin(y), \cos(x), \cos(y) )$ 
with the abstract system having the dynamics
\[\begin{array}{rcl}
 \frac{dw_1}{dt} &=& w_3 w_6 + w_4 w_5 \\
 \frac{dw_2}{dt} &=& w_1 + w_2 \\
 \frac{dw_3}{dt} &=& w_3 w_5 w_6 + w_5^2 w_4 \\
 \frac{dw_4}{dt} &=& w_6 w_1 + w_6 w_2 \\
 \frac{dw_5}{dt} &=& - w_3^2 w_6 - w_3 w_4 w_5 \\
\frac{dw_6}{dt} &=& - w_4 w_1 - w_4 w_2 \\
\end{array}\]
Here $w_1,\ldots,w_6$ correspond to the components of the map $\alpha$ above.\hfill\halmos
\end{example}

\begin{theorem}\label{Theorem:thm-conv}
Given an initial vector space $V_0$ and vector field $\F$, the
iterative procedure using Eqn.~\eqref{Eq:refineV} converges in finitely
many steps to a subspace $V^* \subseteq V_0$. Let
$\alpha_1,\ldots,\alpha_m$ be the basis functions that generate
$V^*$.  \begin{enumerate}
\item\label{Stmt:1}  The transformation $\alpha: (\alpha_1,\ldots,\alpha_m)$ generated by
  the basis functions of the final vector space leads to an abstract system
  whose dynamics are described by polynomials of degree at most $d$.

\item   For every  CoB transformation $\beta:
  (\beta_1, \ldots,\beta_k)$, wherein each $\beta_i \in V_0$ and
    $\beta$ yields a polynomial abstraction of degree at most $d$, it follows
    that $\beta_i \in V^*$.
\end{enumerate}
\end{theorem}
\begin{proof}
  Let us represent the iterative sequence as
  \[ V_0 \supseteq V_1 \supseteq V_2 \cdots \]
  The convergence of the iteration follows from the observation that
  if $V_{i+1} \subset V_{i}$, the dimension of $V_{i+1}$ is at least
  one less than that of $V_i$. Since $V_0$ is finite dimensional, the
  number of iterations is upper bounded by the number of basis
  functions in $V_0$.

  Statement~\ref{Stmt:1} follows directly from Theorem~\ref{Thm:d-closed-thm}.

  Finally, us assume that a transformation $\beta$ exists such that
  $\beta_i \in V_0$. We note that the space $U$ generated by
  $\mathbf{1}, \beta_1,\ldots,\beta_k$ is a subset of $V_0$ and is
  $d-$closed.  We can now prove by induction that $U \subseteq V_i$
  for each $i$.  The base case is true since $U \subseteq V_0$. 

   Next, we show that if $U \subseteq V_i$ then $U \subseteq
  V_{i+1}$. This follows from Eq.~\ref{Eq:refineV} since for each
  $p \in U$, we have $p \in V_i$ and $\lie(p) \in \pp{U}{d}$. This
  gives us $\lie(p) \in \pp{V_i}{d}$.  Therefore, $p \in V_{i+1}$.

As a result, we prove by induction that $U \subseteq V_i$ for each
   $i$.  This also means that $U \subseteq V^*$.
\end{proof}

Note that it is possible for the converged result $V^*$ to be
trivial. I.e, it is generated by the constant function $\mathbf{1}$.
\begin{example} 
Consider the Vanderpol oscillator whose dynamics are given by 
\[ \dot{x} = y,\ \dot{y}= \mu( y - \frac{1}{3} y^3 - x) \,.\]
Our search for polynomials ($\mu = 1$) of degree up to 20 did not yield a
non-trivial linearizing transformation.
\end{example}
For a trivial system, the resulting affine system $\scrT$ is
$\frac{dy}{dt} =0$ under the map $\alpha(\vx)=0$. Naturally, this
situation is not quite interesting but will often result, depending on
the system $\scrS$ and the initial basis chosen $V_0$. We now discuss
common situations where the vector space $V^*$ obtained as the result
is guaranteed to be
non-trivial.

\subsection{Strong and Constant Scale Consecution}

The notion of ``strong'' consecution, ``constant scale'' consecution
and ``polynomial scale'' consecution were defined for equality
invariants of differential equations in our previous
work~\cite{Sankaranarayanan+Sipma+Manna/2008/Constructing} and
subsequently expanded upon by Matringe et
al.~\cite{Matringe+Others/2009/Morphisms} using the notion of
morphisms.  We now show that the techniques presented in this section
can generalize strong and constant scale consecutions, ensuring that
all the systems handled by the techniques presented in our previous
work~\cite{Sankaranarayanan+Sipma+Manna/2008/Constructing} can be
handled by the techniques here (but not vice-versa).

\begin{definition}
A function $f$ satisfies the \emph{strong scale} consecution
requirement for a vector field $\F$ iff $\lie_F(f) = 0$. In
other words, $f$ is a conserved quantity.
Similarly, $f$ satisfies the \emph{constant scale} consecution  
iff $\exists \lambda \in \reals,\ \lie_F(f) = \lambda f$.
\end{definition}
The following theorem is a corollary of Theorem~\ref{Theorem:thm-conv}
and shows that the ideas presented in this section can capture
the notion of strong and constant scale consecution without
requiring quantifier elimination, solving an eigenvalue problem~\cite{Sankaranarayanan+Sipma+Manna/2008/Constructing} or finding roots of a univariate
polynomial~\cite{Matringe+Others/2009/Morphisms}.
\begin{theorem}
  The result of the iteration $V^*$ starting from an initial space
  $V_0$ contains all the strong and constant scale invariant functions
  in $V_0$.
\end{theorem}
\begin{proof}
This is a direct consequence of Theorem~\ref{Theorem:thm-conv} by noting that
for a constant scale consecuting function $f$, the subspace $U \subseteq V_0$
spanned by  $f$ is closed under Lie derivatives.
\end{proof}

Furthermore, if such functions exist in $V_0$ the result after
convergence $V^*$ is guaranteed to be a non-trivial vector space (of
positive dimension). Finally, constant scale and strong scale
functions can be extracted by computing the affine equality invariants
of the linear system $\scrT$ that can be extracted from $V^*$.

\subsubsection{Stability}
We briefly address the issue of deducing stability (or instability) of
a system $\scrS$ using an abstraction to a system $\scrT$. Since
$\alpha$ satisfies the identity
\[ \G(\alpha(\vx)) = J_{\alpha}. \F(\vx) \,.\]
Every equilibrium of $\scrS$ ($\F(\vx) = 0$) maps onto an equilibrium
of $\scrT$ ($\G(\vx) = 0$), but not vice-versa. Furthermore, the map
$\alpha(\vx) = (\mathbf{0},\ldots,\mathbf{0})$ is an abstraction from
any non-linear system to one with an equilibrium at origin. Therefore,
unless restrictions are placed on $\alpha$, we are unable to draw
conclusions on liveness properties for $\scrS$ based on $\scrT$.  If
$\alpha$ has a continuous inverse, then $\scrT$ is topologically
diffeomorphic to $\scrS$~\cite{Meiss/2007/Differential}. This allows
us to correlate equilibria of $\scrT$ with those of $\scrS$.  The
preservation of stability under mappings of state variables has been
studied by Vassilyev and
Ul'yanov~\cite{Vassilyev+Ulyanov/2009/Preservation}.  We are currently
investigating restrictions that will allow us to draw conclusions
about liveness properties of $\scrS$ from those of $\scrT$.

The issue of stability preserving maps between continuous and hybrid systems
was recently addressed by the work of Prabhakar et al.~\cite{Prabhakar+Viswanathan/2012/Stability}.

\subsection{Affine CoB Abstraction: Existence}
We will now focus on the special case of CoB transformations that lead
to linear abstractions of the form $\frac{d \vec{w}}{dt} = A \vec{w}$
(and affine abstractions of the form $\frac{d \vec{w}}{dt} = A \vec{w}
+ \vec{b}$).

Let $\scrS$ be a non-linear system over $\vx$ that has a CoB
transformation $\alpha: \reals^n \rightarrow \reals^m$ with $m > 0$
that maps to a linear system $\frac{d \vec{w}}{dt} = A \vec{w}$.  

\begin{lemma}
The system $\scrS$ has $m$ conserved quantities given by the
components of the vector valued function $ e^{-tA} \alpha(\vx)$.
\end{lemma}
\begin{proof}
Our goal is to prove that the Lie derivative of each component of $e^{-tA} \alpha(\vx)$ equals
zero. Since $\alpha$ is a linearizing CoB, we have
$ \lie(\alpha(\vx)) = A \alpha(\vx)$. 

The Lie derivative of $e^{-tA} \alpha(\vx)$ is given by
\[ e^{-tA} \lie(\alpha(\vx)) + \partial_t e^{-tA} \alpha(\vx) = e^{-tA} A \alpha(\vx) - e^{-tA} A \alpha(\vx) = 0 \,.\]
Thus we see that the Lie derivative of $e^{-tA} \alpha(\vx)$
vanishes. Therefore, each component of $e^{-tA} \alpha(\vx)$ is a
conserved quantity.
\end{proof}

Conversely, whenever the original system $\scrS$ has conserved
quantities, it trivially admits the linearization
$\frac{d \vec{w}}{dt} = 0$ using a transformation $\alpha$ that is
formed by its conserved quantity.

\begin{theorem}\label{Thm:existence-affine-trs}
A system $\scrS$ has an independent, linearizing CoB transformation
$\alpha: \reals^n \mapsto \reals^m$ if and only if it has $m$ linearly
independent conserved quantities.
\end{theorem}

The theorem extends to affine CoB transformations that yield abstract
systems of the form $\frac{d\vec{w}}{dt} = A \vec{w} + \vec{b}$.
While conservative mechanical and electromagnetic systems naturally
have conserved quantities (eg., conservation of momentum, energy,
charge, mass), many systems encountered are dissipative. Such cases
are handled by extending the approach presented here to differential
inequality abstractions~\cite{Sankaranarayanan/2011/Automatic}.

Furthermore, even in a setting where conservative quantities exist,
the advantages of searching for a CoB transformation as opposed to
directly searching for a conserved quantity from an ansatz are not
clear at a first glance. The advantage of the techniques presented
here lies in the fact that existing techniques that search for
conserved quantities focus for the most part on finding polynomial
conserved quantities. Whereas, searching for a CoB transformation
allows us to implicitly obtain conserved quantities that may involve
exponentials, sines and cosines in addition to polynomial conserved
quantities by focusing purely on reasoning with vector spaces
generated by polynomials.

\begin{example}
We observed the following conserved quantity for the system in
Example~\ref{Ex:motivating-example}
\[ \begin{array}{l}
\left(  \frac{e^{-9 t}}{51}+\frac{1}{102} \left(50+7 \sqrt{51}\right) e^{\left(-9+\sqrt{51}\right) t}+\frac{1}{102} \left(50-7 \sqrt{51}\right) e^{-\left(9+\sqrt{51}\right) t}\right)\ x +  \\
\left(  -\frac{1}{102} e^{-9 t-\left(9+\sqrt{51}\right) t} \left( 
\begin{array}{l} 
7 e^{9 t}-\sqrt{51} e^{9 t}-14 e^{\left(9+\sqrt{51}\right) t}+\\
7 e^{9 t+\left(-9+\sqrt{51}\right) t+\left(9+\sqrt{51}\right) t}+ \\
\sqrt{51} e^{9 t+\left(-9+\sqrt{51}\right) t+\left(9+\sqrt{51}\right) t}
\end{array}\right)
\right)\ xy + \\
\left( \frac{1}{204} e^{-9 t-\left(9+\sqrt{51}\right) t} \left(e^{9 t}-2 e^{\left(9+\sqrt{51}\right) t}+e^{9 t+\left(-9+\sqrt{51}\right) t+\left(9+\sqrt{51}\right) t}\right) \right) xy^2 
\end{array} \]

This is one of the three conserved quantities obtained by computing $e^{-t A} \alpha(\vx)$, where
\[ \alpha: ( x, xy, xy^2) \ \mbox{and}\ A = \left( \begin{array}{ccc} 2 & 1 & 0 \\ 1 & 9 & \frac{1}{2} \\ 0 & 2 & 16 \\ \end{array}\right) \,.\]
We are unaware of techniques that can directly generate such conserved
quantities. \hfill \halmos
\end{example}

Finally, we conclude by noting that conserved quantities such as the
one described above seem less useful for reasoning about the dynamics
of the underlying system when compared to the CoB transformation and
the resulting abstraction that gave rise to them.

\section{Abstractions for Discrete and Hybrid Systems}\label{Section:prgTransformation}
In this section, we will discuss how the techniques of the previous
sections can be extended to find CoB transformations of purely
discrete programs. In particular, our focus will be on transforming
loops in programs to infer abstractions that are of a simpler
form. Our presentation will first focus on simple loops consisting of
a single location. The combination of loops with multiple locations
and continuous dynamics will be handled in the subsequent section.

\subsection{Transition System Models}

We will first define transition system models and the action of CoB
transformations on these models. Let $\vx \in X$ represent real valued
system variables, where $X \subseteq \reals^n$.  Transition systems
will form our basic models for loops in
programs~\cite{Manna+Pnueli/95/Temporal}.
\begin{definition}\label{Def:transition-system}
  A \emph{transition system} $\Pi$ is defined by a tuple $\tupleof{X,
    L, \T, X_0, \ell_0}$, wherein,
\begin{enumerate}
\item $X \subseteq \reals^n$ represents the continuous state-space. We
  will denote the system variables by $\vx \in \reals^n$.
\item $L$ denotes a finite set of \emph{locations}.
\item $\T$ represents a finite set of \emph{transitions}. Each transition 
  $t_j \in \T$ is a tuple $\tupleof{\ell_j,m_j, G_j, F_j}$, where
\begin{itemize}
\item $\ell_j \in L$ is the pre-location of the transition, and $m_j
  \in L$ is the post-location.
\item $G_j \subseteq \reals^n$ is the guard condition on the system variables $\vx$.
\item $F_j: \reals^n \rightarrow \reals^n$ is the update function.
  \end{itemize}
\item $X_0 \subseteq X $ represents the possible set of initial values
  and $\ell_0 \in L$ represents the starting location.
\end{enumerate}
\end{definition}

\begin{figure}[t]
\begin{center}
\begin{tikzpicture}
\draw (2.5,.5) node[circle,fill=blue!20,draw=black](n0){$\ell_0$};
\path[->] (n0) edge[loop above]node[right]{$t_1$} (n0);
\path[->] (n0) edge[loop right]node[right]{$t_2$} (n0);
\draw (0,-1) node {$\begin{array}{rclp{.5cm}rcl}
\vx &:& (x,y,k) \\
L &:& \{ \ell_0\} \\
\T &:& \left\{ \begin{array}{l}
t_1:\ (\ell_0,\ell_0,G_1,F_1),\\
t_2:\ (\ell_0,\ell_0, G_2, F_2)
\end{array}\right \} \\
X_0 &:& \{ (x,y,k)\ |\ x = y = 0\ \land\ k > 0 \} \,. \\[5pt]
G_1 &:& \{ (x,y,k)\ |\ y < k \}  && G_2 &:& \{ (x,y,k)\ |\ y \geq k \} \\
F_1 &:& \lambda (x,y,k).\ (x+y^2, y+1,k) &&  F_2 &:& \lambda (x,y,k).\ (x,y,k)\\

\end{array}$};

\end{tikzpicture}
\end{center}
\caption{Transition system model for the loop in Example~\ref{Ex:motivating-example-discrete}.}~\label{Fig:prg-trs-model-ex}
\end{figure}
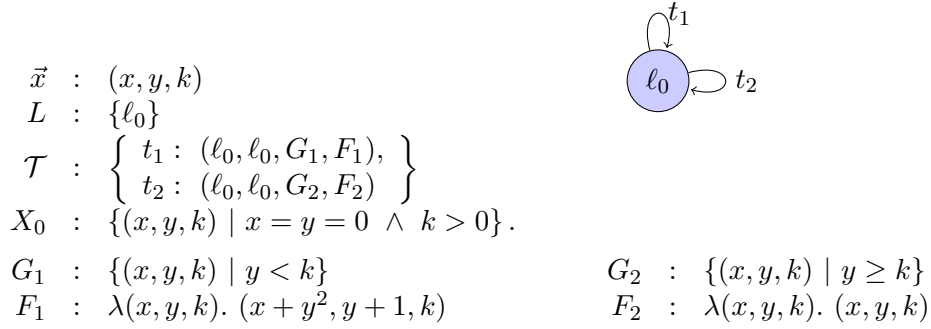

\begin{example}\label{Ex:prg-trs-model-ex}
  Figure~\ref{Fig:prg-trs-model-ex} shows an example of a transition
  system derived from a simple program that computes the sum of the
  first $k$ squares. The transition system consists of a single
  location $\ell_0$, transitions $t_1: (\ell_0,\ell_0,G_1,F_1)$ and
  $t_2: (\ell_0,\ell_0, G_2, F_2)$.  \hfill \halmos
\end{example}

A \emph{state} of the transition system is a tuple
$\sigma: \tupleof{\ell,\vx}$ where $\ell$ is the \emph{current} location and
$\vx \in X$ are the values of the continuous variables.

A \emph{run} is a finite or infinite sequence of states
\[ \sigma_0 \xrightarrow{t_0} \sigma_1 \xrightarrow{t_1} \cdots
\rightarrow \sigma_j \xrightarrow{t_j} \sigma_{j+1} \cdots \,,\] where
each $\sigma_j: (\ell_j,\vx_j)$ is a state and $t_j$ a transition,
satisfying the following conditions:
\begin{enumerate}
\item The starting state $\sigma_0: (\ell_0, \vx_0)$ is initial. I.e.,
  $\ell_0$ is the initial location of $\Pi$ and $\vx_0 \in X_0$.
\item The state $\sigma_{i+1}: (\ell_{i+1},\vx_{i+1})$ is related to the state 
$\sigma_i: (\ell_i, \vx_i)$ in the following way:
\begin{enumerate}
\item The transition $t_i \in \T$ is of the form $(\ell_i,\ell_{i+1},
  G_i,F_i)$, leading from $\ell_i$ to $\ell_{i+1}$.
\item The valuation  $\vx_i$  of the continuous variables 
satisfy the guard $G_i$ and the valuation $\vx_{i+1}$ is obtained
by executing the assignments in $F_i$ on $\vx_i$:
\[ \vx_{i} \in G_i\; \mbox{and}\; \vx_{i+1} = F_i(\vx_i) \,.\]
\end{enumerate}

\end{enumerate}

A special class of ``simple loop'' transition systems
that have a single location are defined below. 
\begin{definition}\label{Def:simple-loop}
A transition system $\Pi$ is called a \emph{simple loop} if it has
a single location. I.e., $L = \{ \ell\}$. All transitions of a simple
loop are self-loops around this location $\ell$.
\end{definition}

The transition system in Example~\ref{Ex:prg-trs-model-ex} is a simple
loop. It consists of a single location. In general, simple loops can
have multiple transitions that ``loop'' around this single location.

We will now discuss the pre-image operator $\fpre$
induced by a transition. Let $g(\vx)$ be some function over
the state variables and $t:\ (\ell,m,G,F)$ be a transition.
\begin{definition}\label{Def:fpre}
  The \emph{functional pre-image} $\fpre(g,t)$ is defined as $
  g(F(\vx))$.
\end{definition}

\paragraph{Note:} The standard precondition operator works over
assertions over the state variables, involving computing the pre-image
using $F$ and computing the intersection of the result with the
guard. The functional precondition defined here is defined over
functions $g(\vx)$ over the state variables.

\begin{example}
Consider the transition 
\[t: ( \ell,m, G,F),\ \mbox{wherein}\  G:\ \{ (x,y)\ |\ x \geq y \}, F:\ \lambda (x,y).\ (x^2, y^2 - x^2 ) ) \,.\]
The functional pre-image of the function $g(x,y): x+y$,
denoted $\fpre(x+y,t)$, is given by
\[ \fpre(x+y,t):\ (x^2) + ( y^2 - x^2) = y^2 \,.\] 

To contrast with the standard pre-condition operator, which applies to
assertions over states, let us consider the assertion $x+y \geq 0$.
We have 
\[ \pre(x+y \geq 0, t):\ y^2 \geq 0\ \land\ x \geq y \,.\]
\hfill\halmos
\end{example}

We now show that $\fpre$ is a linear operator over functions.
\begin{lemma}\label{Lemma:linear-fpre}
  For any transition $t$ and functions $g_1,g_2,g$ over $\vx$, we have
  $\fpre(g_1+g_2,t) = \fpre(g_1,t) + \fpre(g_2,t)$ and further,
  $\fpre(\lambda g) = \lambda \fpre(g)$ for any $\lambda \in \reals$.
\end{lemma}
\begin{proof}
Proof follows by directly applying Def.~\ref{Def:fpre}.
\end{proof}

Let us consider any run of the transition system
\[r: \sigma_0 \xrightarrow{t_0} \sigma_1 \rightarrow \cdots
\rightarrow \sigma_i \xrightarrow{t_i} \sigma_{i+1} \cdots \,.\] Let
$t_i:(\ell_i,\ell_{i+1}, G_i, F_i)$ denote the transition between
$\sigma_i:(\ell_i,\vx_i)$ and $\sigma_{i+1}:(\ell_{i+1},
\vx_{i+1})$. Finally, let $g(\vx)$ be any function over the state
variables of the transition system.
\begin{lemma}
  The following identity holds for all successive pairs of states $
  (\ell_i, \vx_i) \xrightarrow{t_i} ( \ell_{i+1}, \vx_{i+1})$
  encountered in a run of the transition system and for all functions
  $g(\vx)$:
 \[ \fpre(g, t_i ) (\vx_i) \equiv   g(\vx_{i+1}) \,\]
\end{lemma}
\begin{proof}
  We may write $\fpre(g, t_i ) (\vx_i) = g(F(\vx_i))$.  We know that
  $\vx_{i+1} = F(\vx_i)$. Therefore, $g(\vx_{i+1}) = g(F(\vx_i)) =
  \fpre(g, t_i ) (\vx_i) $.
\end{proof}

We will now discuss change-of-basis abstractions for transition
systems.  The discussion will focus on defining change-of-basis
abstractions for simple loops, which are represented by a transition
system with a single location $\ell$
(Cf. Definition~\ref{Def:simple-loop}). The subsequent sections will
extend this concept to arbitrary transition systems.

\subsection{CoB Abstractions For Simple Loops}

Consider a simple loop $\Pi$ over $\vx \in \reals^n$ with a single
location $\ell$, transitions $\{t_1,\ldots,t_k\}$, and initial
condition $X_0$. We seek to abstract $\Pi$ with another simple loop
$\Xi$ over $\vy \in \reals^l$ with a single location $m$, transitions
$\{t_1',\ldots,t_k'\}$ and initial condition $Y_0$.

\begin{definition}\label{Def:cob-abstraction-simple-loops}
  Simple loop $\Xi$ is a CoB abstraction of $\Pi$ iff there is a
  continuous function $\alpha: \reals^n \rightarrow \reals^l$ such
  that
\begin{enumerate}
\item The initial condition $Y_0 \supseteq \alpha(X_0)$,
\item For each transition $t_i: (\ell,\ell,G_i,F_i)$ in $\Pi$, there
is a corresponding transition $t_i': (m,m, G_i', F_i')$ in $\Xi$ such that
\begin{enumerate}
\item $G_i' \supseteq \alpha(G_i)$,
\item $\forall\ \vx\ F_i'(\alpha(\vx)) =  \alpha(F_i(\vx))$.
\end{enumerate}
\end{enumerate}
\end{definition}

We will now present an example of CoB abstraction for simple loops.
\begin{example}\label{Ex:prg-abstraction}
  Consider the simple loop from Example~\ref{Ex:prg-trs-model-ex}
  (also Fig.~\ref{Fig:prg-trs-model-ex}). We note that the map
  \[ \alpha: \reals^3 \rightarrow\ \reals^4,\ \mbox{where}\ \alpha=
  \lambda(x,y,k).  (x, y, k, y^2) \,, \] yields an abstract transition
  system $\Xi$ over variables $\vw: (w_1,w_2,w_3,w_4)$. Informally,
  the variables $(w_1,w_2,w_3,w_4)$ are place holders for the
  expressions $(x,y,k,y^2)$, respectively. The resulting transition
  system $\Xi$ is

  \[ \begin{array}{rcl}
    \vw &:& (w_1,\ldots,w_4) \\
    L &:& \{ m \} \\
    \T &:& \{t_1':\ (m,m,G_1', F_1') ,t_2': (m,m,G_2', F_2')\} \\
    X_0 &:& w_1 = w_2 = w_4 = 0\ \land\ w_3 \geq 1 \\
    G_1' &:& \{ \vw\ |\ w_2 < w_3 \} \\
    G_2' &:& \{ \vw\ |\ w_2 \geq w_3 \}\\
    F_1' &:& \lambda\ \vw.\ (w_1 + w_4, w_2+1,w_3, w_4 + 2 w_2 + 1) \\
    F_2'  &:& \lambda\ \vw.\ \vw \\
  \end{array}\]
  The various requirements laid out in Definition~\ref{Def:cob-abstraction-simple-loops}
  can be easily verified. We will verify the requirement for $F_1'$: 
  $F_1'(\alpha(x,y,k)) = \alpha(F_1(x,y,k))$, as follows:
\[ \begin{array}{rclcl}
F_1'(\alpha(x,y,k)) &=& F_1'(x,y,k,y^2) &=& (\underset{w_1+w_4}{\underbrace{x+y^2}}, \underset{w_2+1}{\underbrace{y+1}}, \underset{w_3}{\underbrace{k}}, \underset{w_4+2w_2 +1 }{\underbrace{y^2+2y +1}}) \\[5pt]
&=& \alpha(x+y^2,y+1,k) &=& \alpha(F_1(x,y,k))\\
\end{array} \,.\]
\hfill\halmos
\end{example}

The definition of CoB abstraction immediately admits the following key theorem.
\begin{theorem}\label{Thm:run-correspondence}
For any run 
\[ \sigma_0: (\ell,\vx_0) \xrightarrow{t_0} (\ell, \vx_1) \xrightarrow{t_1} (\ell,\vx_2) \xrightarrow{t_2} \cdots \]
the corresponding sequence of $\Xi$-states
\[ \gamma_0: (m, \alpha(\vx_0)) \xrightarrow{t_0'} (m,\alpha(\vx_1)) \xrightarrow{t_1'} (m, \alpha(\vx_2)) \xrightarrow{t_2'} \cdots \,,\]
is a run of $\Xi$.
\end{theorem}
\begin{proof}
  Proof uses the property that whenever the move $(\ell,\vx_j)
  \xrightarrow{t_j} (\ell,\vx_{j+1})$ is enabled in $\Pi$ then the
  move $(m,\alpha(\vx_j)) \xrightarrow{t_j'} (m,\alpha(\vx_{j+1}))$ is
  enabled  in $\Xi$.

  Let $t_j$ be described by the guard $G_j$ and the functional update
  $F_j$.  Likewise, let $t_j'$ be described by $G_j'$ and $F_j'$.  We
  note that $\alpha(G_j) \subseteq G_j'$. Since $\vx_j$ satisfies the
  guard of $t_j$, $\alpha(\vx_j)$ satisfies that of $t_j'$.  The state
  obtained after the transition is given by 
\[  F'(\alpha(\vx_j)) =  \alpha(F(\vx_{j})) = \alpha(\vx_{j+1}) \,.\]

  We have proved that whenever the move $(\ell,\vx_j)
  \xrightarrow{t_j} (\ell,\vx_{j+1})$ is possible in $\Pi$ then the
  move $(m,\alpha(\vx_j)) \xrightarrow{t_j'} (m,\alpha(\vx_{j+1}))$ is
  possible in $\Xi$. The rest of the proof extends this to trace
  containment through induction over prefixes of the traces.
\end{proof}

As a direct consequence, we may state a theorem that corresponds to
Theorem~\ref{Thm:weak-preserve-inv} for the case of vector fields.
\begin{theorem}
  Let $\denotation{\varphi}$ be an invariant set for the abstract system
  $\Xi$. Then, $\alpha^{-1}(\denotation{\varphi})$ is an invariant of the original
  system $\Pi$.
\end{theorem}
\begin{proof}
  First, we note from Theorem~\ref{Thm:run-correspondence} that if
  $(\ell,\vx)$ is reachable in $\Pi$ then $(m,\alpha(\vx))$ is
  reachable in $\Xi$.  Since $\varphi$ is an invariant for $\Xi$, we
  have $(m,\alpha(\vx)) \in \denotation{\varphi}$. Therefore for any
  reachable state $(\ell,\vx)$ in $\Pi$, we have $ (\ell,\vx) \in 
  \alpha^{-1}(\denotation{\varphi})$. Thus
  $\alpha^{-1}(\denotation{\varphi})$ is an invariant set for $\Pi$.
\end{proof}

Given an invariant $\varphi[\vy]$ for $\Xi$ in the form of an
assertion, the invariants for the original system are obtained simply
by substituting $\alpha(\vx)$ in the place of $\vy$ in $\varphi$.

\begin{example}
Consider the transition system $\Pi$ from Example~\ref{Ex:prg-trs-model-ex}
and its abstraction $\Xi$ in Example~\ref{Ex:prg-abstraction}.  We note
that $\Xi$ has affine guards and updates. Therefore, we may use a standard
polyhedral analysis tool to compute invariants over $\Xi$~\cite{Cousot+Halbwachs/78/Automatic,Halbwachs+Proy/97/Verification,Sankaranarayanan+Colon+Sipma+Manna/2006/Efficient}. Some of the invariants obtained include
\[ \begin{array}{l}
  13 w_4 \leq  9 w_1 +24 w_2  \ \land\   7 w_4 \leq  6 w_1 +11 w_2  \ \land\  4 w_1 +7 w_2 - 7 w_4 +11 w_3  \geq  11 \\
  2 w_1 +3 w_2 - 3 w_4 +4 w_3  \geq  4 \ \land\   w_4 \leq  2 w_1 + w_2   \ \land\  3 w_4 \leq   w_1 +12 w_2  \\
  9 - w_1 -3 w_2 + 3 w_4 -9 w_3  \leq  0 \ \land\  w_2 \geq  0 \ \land\  1  \leq  w_3 \ \land\ 
  w_2 - w_3  \leq  0\\
\end{array}\]
By substituting $w_1 \mapsto x, w_2 \mapsto y, w_3 \mapsto k, w_4 \mapsto y^2$  on these invariants, we conclude invariants
for the original system. For instance, we conclude facts such as
\[ 13 y^2 - 24 y - 9 x \geq 0\ \land\ 7 y^2 - 11 y - 6x \geq 0\ \land\ 11 k - 7 y^2 + 7 y + 4 x \geq 11 \,.\]
\hfill\halmos
\end{example}

The goal, once again, is to find an abstraction $\alpha$ and an
abstract system $\Xi$ starting from a description of the system $\Pi$.
Furthermore, we require that the update functions $F_j'$ in $\Xi$ are
all polynomials whose degrees are smaller than some given limit $d >
0$.  In particular, if we set $d = 1$, we are effectively requiring
all the updates in $\Xi$ to be affine functions over $\vy$.
 
Our strategy will be to find a map $\alpha: \reals^n \rightarrow
\reals^k$.  For convenience, we will write $\alpha$ as
$(\alpha_1,\ldots,\alpha_k)$, wherein each component function
$\alpha_j: \reals^n \rightarrow \reals$.  Let $V$ be the vector space
spanned by the components of $\alpha$, i.e, $V = \vspan(\{
\alpha_1,\ldots,\alpha_k \})$.  Our goal will be to ensure that for
each transition $t$ in $\Pi$ and for each $\alpha_i$, 
\begin{equation}\label{Eq:vector-space-closure-loops}
 \forall\ \vx,\ \fpre(\alpha_i(\vx),t)\ \in\ \pp{V}{d} \,.
\end{equation}

Let $V$ be a vector space that satisfies Eq.~\eqref{Eq:vector-space-closure-loops}
for each transition $t$ in $\Pi$. We will say that the space 
$V$ is \emph{$d$-closed} w.r.t $\Pi$.

\begin{theorem}
  Let $V: \vspan(g_1,\ldots,g_k)$ be $d$-closed w.r.t $\Pi$ for
  continuous functions $g_1,\ldots,g_k$.  The map $\alpha:
  (g_1,\ldots,g_k)$ is a CoB transformation defining an abstract
  system $\Xi$, wherein each transition of $\Xi$ has a polynomial
  update function involving polynomials of degree at most $d$.
\end{theorem}
\begin{proof}
  We construct the abstract system $\Xi$ with variables
  $w_1,\ldots,w_k$ representing the functions $g_1,\ldots,g_k$ that
  are the components of $\alpha$.  $\Xi$ has a single location $m$ and
  for each transition $t_i \in \Pi$, we construct a corresponding
  transition $t_i' \in \Xi$ as follows.

  Let $G_i, F_i$ be the guard set and update function for $t_i$,
  respectively. The guard set for $t_i'$ is given by $\alpha(G_i)$ or
  an over-approximation thereof. Likewise, the update $F_i'$ for
  $t_i'$ is derived as follows. We note that
\[ \fpre(g_j, t_i) = \sum_{r}
  c_{r_1,r_2,\ldots,r_k} g_1^{r_1} g_2^{r_2} \cdots g_k^{r_k}\,,\]
 wherein  $0 \leq r_1 + r_2 + \ldots + r_k \leq d$. The corresponding update
 for $w_j$ in the abstract system is given by
 \[ F_i' (w_j) = \sum_{r} c_{r_1,r_2,\ldots,r_k} w_1^{r_1} w_2^{r_2}
 \cdots w_k^{r_k}\,.\] 

Note that each function $F_i'(w_j)$ is a
 polynomial of degree at most $d$ over $w_1,\ldots,w_k$.
\end{proof}

Since the operator $\fpre$ used to define the closure in
Eq.~\eqref{Eq:vector-space-closure-loops} is a linear operator
(Cf. Lemma~\ref{Lemma:linear-fpre}), we may check the closure property
for a given vector space $V$ by checking if its basis functions
satisfy the property.
\begin{lemma}
  The vector space $V: \vspan(\{g_1,\ldots,g_k\})$ is $d$-closed w.r.t
  $\Pi$ iff for each basis element $g_i$ of $V$, and for each
  transition $t$ in $\Pi$, $\fpre(g_i,t) \in \pp{V}{d}$.
\end{lemma}
\begin{proof}
  For the non-trivial direction, let $V$ be a space where for each
  basis element $g_i$ of $V$, and for each transition $t$ in $\Pi$,
  $\fpre(g_i,t) \in \pp{V}{d}$. An arbitrary element $g \in V$ can be
  written as a linear combination of its basis elements: $g = \sum_j
  \lambda_j g_j$.  We have $\fpre(g,t) = \sum_j \lambda_j
  \fpre(g_j,t)$ from Lemma~\ref{Lemma:linear-fpre}.  Since
  $\fpre(g_j,t) \in \pp{V}{d}$, which is a vector space itself, we
  have that $\fpre(g,t)$ is a linear combination of elements in
  $\pp{V}{d}$ and thus $\fpre(g,t) \in \pp{V}{d}$. Thus $V$ is
  $d$-closed.
\end{proof}

\begin{example}
  Once again, consider the system $\Pi$ in
  Example~\ref{Ex:prg-trs-model-ex} and the map $\alpha: (x,y,k,y^2)$
  from Example~\ref{Ex:prg-abstraction}.  The components of this map
  are the functions $\alpha_1: x, \alpha_2: y, \alpha_3: k,
  \mbox{and}\ \alpha_4: y^2$. We may verify that the vector space $V:
  \vspan(\{x,y,k,y^2\})$ satisfies the closure property in
  Eq.~\eqref{Eq:vector-space-closure-loops} for $d=1$.  The table
  below shows the results of applying $\fpre$ on each of the basis
  elements.
  \[ \begin{array}{|l|l|l|}
\hline
    \mbox{Basis function}\ g_j & \fpre(g_j,t_1) & \fpre(g_j,t_2) \\ 
\hline
                 x            & x + y^2        & x \\
                 y            & y +1           & y \\
                 k            & k              & k \\
                 y^2          & y^2 + 2y +1    & y^2 \\
\hline
\end{array}\]
Thus, $\fpre(g_j,t_k)$ belongs to $\pp{V}{1} = \vspan(\{1,x,y,k,y^2\})$. \hfill\halmos

\end{example}

\paragraph{Searching for Abstractions:} The procedure for finding
abstractions is identical to that used for vector fields with the
caveat that closure under Lie-derivative is replaced by closure under
$\fpre(\cdot,t_j)$ for every transition $t_j$ in the system. The
procedure takes as input an initial basis of functions $B_0$ and
iteratively refines the vector space $V_i: \vspan(B_i)$ by removing
all the functions that do not satisfy the closure property.

\begin{example}
Consider the system $\Pi$ in Example~\ref{Ex:prg-trs-model-ex} and the
initial basis consisting of all monomials of degree at most $2$ over
variables $x,y,k$. We obtain the basis $B_0: \{ x,y,k, x^2, y^2, k^2,
xy, yk, xk \}$ and the space $V_0: \vspan(B_0)$. An element of 
$V_0$ can be written as 
\[ p:\ \left[\begin{array}{c}
c_1 x + c_2 y + c_3 k + c_4 x^2 + c_5 y^2 + c_6 k^2\\
 + c_7 xy + c_8 yk + c_9 xk\,.
\end{array}\right]\]
We consider the transition $t_1$ with update $F_1: \lambda (x,y,k). (x
+ y^2 , y+1, k)$. Transition $t_2$ is ignored as its update is simply
the identity relation. We have $\fpre(p,t_1)$ as
\[ \fpre(p,t_1):\ \left[ \begin{array}{c}
(c_2+ c_5) + (c_1 + c_7) x  + (c_2 + 2 c_5) y + (c_3 + c_8) k + c_4 x^2 +\\
 (c_1 + c_5 +c_7) y^2 + c_6 k^2 + c_7 xy + c_7 y^3 + c_4 y^4 + 2 c_4 xy^2 + \\
c_8 yk +  c_9 xk + c_9 y^2 k 
\end{array}\right] \]
The ``overflow'' terms $c_7 y^3$, $c_4 y^4$, $c_9 y^2k$ immediately
yield the constraints $c_4 = c_7 = c_9 = 0$.  The refined basis is $
B_1: \{ x,y, k, y^2, k^2, yk \}$. The iterative process converges with
$V_1: \vspan(B_1)$ yielding a linearization. \hfill\halmos
\end{example}

\subsection{Abstractions for General Transition Systems}
Thus far, we have presented CoB abstractions for simple loops
consisting of a single location. The ideas seamlessly extend to
systems with multiple locations with a few generalizations that will
be described in this section.

Let $\Pi$ be a system with a set of locations
$L=\{\ell_1,\ldots,\ell_k\}$ and transitions $\T$. We will assume that
$|L| \geq 2$ so that the system is no longer a simple loop.  The main
idea behind change of basis (CoB) transformations for systems with
multiple locations is to allow a different map for each location. In
other words, the abstraction is defined by a maps $\alpha_\ell(\vx)$
for each location $\ell \in L$. 

The maps for two different locations $\ell_1$ and $\ell_2$ are of the
type $\alpha_{\ell_1}: \reals^n \rightarrow \reals^{m_1}$ and
$\alpha_{\ell_2}: \reals^n \rightarrow \reals^{m_2}$. In general, we
may assume that $m_1 \not= m_2$. This discrepancy can be remedied by
padding each $\alpha_{\ell_i}$ with extra components that map to the
constant function $0$. While, this transformation violates the linear
independence requirement between the various components in $\alpha$,
it makes the resulting abstract system easier to describe. Without
loss of generality, we assume that all the maps $\alpha_\ell$ for each
$\ell \in L$ are of the form $\alpha_\ell:\ \reals^n \rightarrow
\reals^m$ for a fixed $m >  0$.

\begin{definition}\label{Def:cob-abstraction-general}
  A system $\Xi$ is a CoB abstraction of $\Pi$ through a collection of
  maps $\alpha_{\ell_1}, \ldots,\alpha_{\ell_k}$ each of the type
  $\reals^n \rightarrow \reals^m$, corresponding to locations
  $\ell_1,\ldots,\ell_k$, iff
\begin{enumerate}
\item $\Xi$ has locations $m_j$ corresponding to $\ell_j \in L$ for
$1 \leq j \leq k$ , and transitions $t_i'$ corresponding to transition 
$t_i \in \T$.
\item For each transition $t_i: \tupleof{\ell_{pre},\ell_{post}, G_i, F_i}$ in 
$\Pi$, the corresponding transition $t_i': \tupleof{m_{pre}, m_{post}, G_i', F_i'}$
is such that 
\begin{enumerate}
\item   $m_{pre}$ and $m_{post}$ correspond to $\ell_{pre}$ and $\ell_{post}$, respectively, 
\item $G_i' \supseteq \alpha_{\ell_{pre}}(G_i)$, 
\item $ (\forall\ \vx)\ F_i'(\alpha_{\ell_{pre}} (\vx)) =  \alpha_{post}(F_i(\vx))$.
\end{enumerate}
\end{enumerate}
\end{definition}

We note that for a simple loop with a single location, the definition
above is identical to Def.~\ref{Def:cob-abstraction-simple-loops}.
\begin{figure}[t]
\begin{tabular}{cc}
\begin{minipage}{5cm}
{\small
\begin{lstlisting}
    int  x,y,z;
    // .. initialize.. 
    #\label{pt-while}#while (x + y - z <= 100){
     (x,y):=( x + z * (x - y) , 
              y + z * (y - x));
     // x,y,z unmodified here 
     (x,y,z) := (z+1 , 
                 x+y -1 , 
                 z+x+y -1 );
     }
\end{lstlisting}
}
\end{minipage} & 
\begin{minipage}{4cm}
\begin{center}
\begin{tikzpicture}
\begin{scope}[xshift=1.5cm,yshift=2cm]
\matrix[every node/.style={circle, fill=blue!30, draw=black, inner sep=2pt}, row sep=10pt, column sep=10pt]{
  \node[inner sep=0pt, minimum size=0pt](p0){}; & \node(n1){\scriptsize $\ell_1$}; & \node[inner sep=0pt, minimum size=0pt](p1){};\\
\node[inner sep=0pt, minimum size=0pt](p2){};  & \node(n2){\scriptsize$\ell_2$};  &  \node(n3){\scriptsize$\ell_3$};\\
};
\path[->] (n1) edge node[right]{$t_1$} (n2)
(n1) edge[-] node[above]{$t_3$}  (p1) 
(p1) edge (n3)
(n2) edge [-] (p2)
(p2) edge [-] node[left]{$t_2$} (p0)
(p0) edge (n1);
\end{scope}
\begin{scope}
\draw (0,0) node {
$\begin{array}{rcl}
  L &:& \{ \ell_1,\ell_2,\ell_3\} \\
  \T &:& \{ t_1, t_2, t_3 \} \\
  t_1 &:& \tupleof{\ell_1, \ell_2, G_1, F_1 } \\
  t_2 &:& \tupleof{\ell_2, \ell_1, G_2, F_2 } \\
  t_3 &:& \tupleof{\ell_1, \ell_3, G_3, F_3 } \\[5pt]
  G_1 &:& \{ (x,y,z)\ |\ x + y -z \leq 100 \} \\
  F_1 &:& \left( 
     x + z x - zy,
     y + z y - z x,
      z    \right) \\[5pt]
  G_2 &:& \mathbb{Z}^3\\
  F_1 &:&  \left( z+1, x+y-1, z +x + y -1 \right) \\[5pt]
  G_3 &:& \{ (x,y,z)\ |\ x + y -z >  100 \} \\
  F_3 &:&  (x,y,z)
 \end{array}$};
\end{scope}
\end{tikzpicture}
\end{center}
\end{minipage}
\end{tabular}
\caption{ An example program fragment with multiple locations and its transition system.}\label{Fig:trs-example-2}

\end{figure}
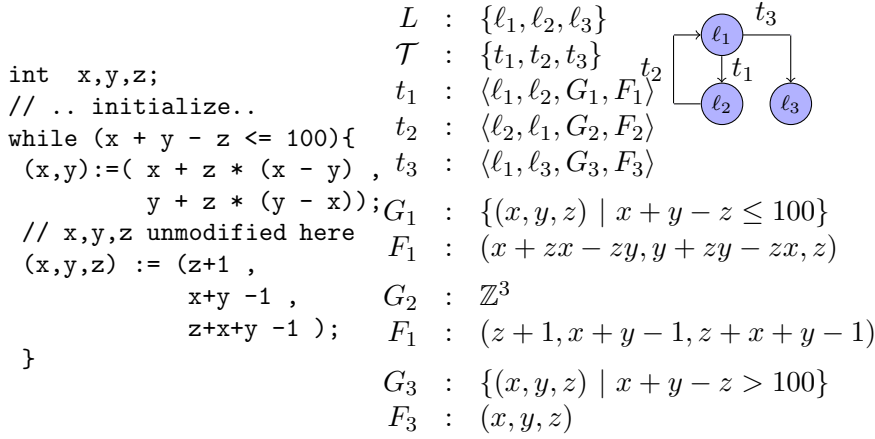
\begin{example}\label{Ex:trs-example-2}
Figure~\ref{Fig:trs-example-2} shows an example of a transition system with 
multiple locations.  Consider the following CoB transformation:
\[ \begin{array}{rcl}
\alpha_{\ell_1} &:& ( z^2, yz, xz, z, y^2, xy, y, x^2, x) \\
\alpha_{\ell_2} &:& (z^2, yz+xz, z, y, y^2+2xy+x^2 , x, 0 , 0,0 )\\
\alpha_{\ell_3} &:&  ( z^2, yz, xz, z, y^2, xy, y, x^2, x) \\
\end{array}\]
The transformation yields an abstraction $\Xi$ 
of the original system.  The abstract system has $9$
variables $w_0,\ldots,w_8$. The structure of 
$\Xi$ mirrors that of $\Pi$ with three locations
$m_1,m_2,m_3$ corresponding to $\ell_1,\ell_2,\ell_3$, respectively
and three transitions $t_1',t_2' $ and $t_3'$ corresponding
to $t_1,t_2$ and $t_3$ in $\Pi$. The guards and updates of 
the transition $t_1'$  are
\[\begin{array}{rcl}
 G_1' &:& \{ (w_0,\ldots,w_8)\ |\ w_8+w_6 -w_3 \leq 100 \},\\
F_1' &:& (w_0, w_1 + w_2, w_3, w_1 - w_2 + w_6, w_4 + 2w_5 + w_7, - w_1 + w_2 + w_8, 0, 0, 0 ) \\
\end{array}\]
We verify the key condition that ensures that $t_1'$ is an abstraction of
$t_1$: 
\[ \alpha_{\ell_2}(F_1(x,y,z)) = F_1'(\alpha_{\ell_1}(x,y,z)) \,.\]
The LHS $\alpha_{\ell_2}(F_1(x,y,z)) = \alpha_{\ell_2}( x+ zx - zy, y + zy- zx, z)$ is given by
\[ (z^2, zx + zy, z, y + zy- zx, x^2 + 2xy + y^2, x + zx - zy, 0,0,0 ) \,.\]
The RHS $F_1' (\alpha_{\ell_1}(x,y,z)) = F_1'(  z^2, yz, xz, z, y^2, xy, y, x^2, x)$ is given by 
\[ (z^2, xz + yz, z, y - zx + zy, y^2 + 2xy + x^2, x + zx - zy, 0,0,0
) \,.\] 
The identity of LHS and RHS is thus verified. \hfill\halmos
\end{example}

Our goal once again is to search of a collection of transformations
$\alpha_{\ell}$, for each $\ell \in L$ such that the resulting system
is described by polynomial updates of degree at most $d$. The case
where $d=1$ corresponds to affine updates.  Once again, we generalize
the notion of a $d-$closed vector space.  Consider a collection of
vector spaces $V_{\ell}: \vspan(B_{\ell})$ for each location $\ell \in
L$.

\begin{definition}
  We say that the collection $V_{\ell}, \ell \in L$ is $d-$closed for
  transition system $\Pi$ if and only if for each transition $t_j:
  \tupleof{\ell_{pre}, \ell_{post}, G_j, F_j}$ and for each element $p
  \in V_{post}$, we have $ \fpre(p,t_j) \in \pp{V_{pre}}{d}$.
\end{definition}

The notion of $d-$closed vector spaces can be related to CoB
transformations and resulting abstractions whose updates are defined
by means of polynomials of degree at most $d$.

\begin{theorem}
Let $V_{\ell}, \ell \in L$ be a collection of vector spaces that are
$d-$closed for a system $\Pi$. The basis elements of $V_{\ell}$ yields
a collection of maps $\alpha_\ell,\ \ell \in L$ that relate $\Pi$ to a
CoB abstraction $\Xi$. The update maps of $\Xi$ are all polynomials of
degree at most $d$.
\end{theorem}

\begin{example}
  Consider the transition system described in
  Example~\ref{Ex:trs-example-2} and
  Figure~\ref{Fig:trs-example-2}. We wish to discover an affine
  abstraction for this system automatically. Starting from the initial
  collection of vector spaces that maps each location to the space of
  all polynomials of degree at most $2$ over $x,y,z$, we obtain the
  transformations $\alpha_{\ell_1}, \alpha_{\ell_2}, \alpha_{\ell_3}$
  described in the same example. This yields an abstract system over
  variables $w_0,\ldots,w_8$.
\end{example}

\subsection{Combining Discrete and Continuous Systems}

As a final step, we extend our approach to hybrid systems that combine
discrete and continuous dynamics.  We define hybrid systems briefly
and extend the results from Sections~\ref{Section:covAbstraction} and
\ref{Section:prgTransformation} to address hybrid systems.

\begin{definition}\label{Def:hybrid-system}
  A hybrid system consists of a discrete transition system $\Pi:
  \tupleof{X,L, \T, X_0, \ell_0}$ and a mapping that associates each
  location $\ell_i \in L$ with a continuous subsystem $\scrS_i:
  \tupleof{\F_i, X_{i}}$ over the state-space $X$, consisting of a
  vector field $\F_i$ and location invariant $X_i$.
\end{definition}

A state $\sigma$ of the hybrid system consists of a tuple
$\tupleof{\ell,\vx,T}$ where $\ell \in L$ is the current location,
valuations to the continuous variables $\vx \in X $ and the current
time $T \geq 0$. 

Given a time $\delta \geq 0$, we write $\tupleof{\ell,\vx,T}\
\underset{\delta}{\leadsto}\ \tupleof{\ell,\vy,T+\delta}$ to denote
that starting from state $\tupleof{\ell,\vx,T}$ the hybrid system
\emph{flows} continuously according to the continuous subsystem
$\scrS_{\ell}$ corresponding to the location $\ell$. Likewise, we
write $\tupleof{\ell,\vx,T} \xrightarrow{t_j} \tupleof{\ell', \vx',
  T}$ to denote a \emph{jump} between two states upon taking a
discrete transition $t_j$ from $\ell$ to $\ell'$.  Note that no time
elapses upon taking a jump.

A run $R$ of the hybrid system is given by a countable sequence of
alternating flows (evolution according to the ODE inside a location)
and jumps (discrete transition to a different location) starting
from an initial state:
\[ \sigma_0: \tupleof{\ell_0, \vx_0, 0} \underset{\delta_0}{\leadsto}
\sigma_0': \tupleof{\ell_0,\vy_0, \delta_0} \xrightarrow{t_1}\ \sigma_1:
\tupleof{\ell_1,\vx_1, \delta_0} \underset{\delta_1}{\leadsto}\ \sigma_1':
\tupleof{\ell_1,\vy_0, \delta_0+\delta_1} \xrightarrow{t_2} \cdots \] 

To avoid Zenoness, we require that the summation of the dwell times in
the individual modes $ \sum_{j=0}^{\infty} \delta_j$ diverges.

We now define CoB abstractions for hybrid systems. Our definitions
simply combine aspects of the definition for transition
systems~\ref{Def:cob-abstraction-general} and continuous
systems~\ref{Def:semi-conjugacy}.

A CoB abstraction of the hybrid system is obtained through a
collection of maps $\alpha_{\ell_1}, \ldots, \alpha_{\ell_k}$
corresponding to the locations $\ell_1, \ldots,\ell_k$ of the hybrid
system. It is assumed that by padding with $0$s, we obtain each
$\alpha_{\ell_i}$ as a function $ \reals^n \rightarrow \reals^m$.

\begin{definition}
 A system $\Xi$ is a CoB abstraction of $\Pi$ through a collection of
  maps $\alpha_{\ell_1}, \ldots,\alpha_{\ell_k}$ each of the type
  $\reals^n \rightarrow \reals^m$, corresponding to locations
  $\ell_1,\ldots,\ell_k$, iff
\begin{enumerate}
\item $\Xi$ has locations $m_j$ corresponding to $\ell_j \in L$ for $1
  \leq j \leq k$ , and transitions $t_i'$ corresponding to transition
  $t_i \in \T$. Each location $m_j$ in $\Xi$ has an associated
  continuous system $\scrT_j$.
\item For each corresponding location pair $\ell_j, m_j$, the system
  $\scrT_j$ is a CoB abstraction of $\scrS_j$ through the
  transformation $\alpha_{\ell_j}$.

\item For each transition $t_i: \tupleof{\ell_{pre},\ell_{post}, G_i,
    F_i}$ in $\Pi$, the corresponding transition $t_i':
  \tupleof{m_{pre}, m_{post}, G_i', F_i'}$ are such that
\begin{enumerate}
\item   $m_{pre}$ and $m_{post}$ correspond to $\ell_{pre}$ and $\ell_{post}$, respectively, 
\item $G_i' \supseteq \alpha_{\ell_{pre}}(G_i)$, 
\item $ (\forall\ \vx)\ F_i'(\alpha_{\ell_{pre}} (\vx)) =  \alpha_{post}(F_i(\vx))$.
\end{enumerate}
\end{enumerate}
\end{definition}

Once again, we focus on searching for an abstraction $\Xi$ of a given
hybrid system wherein the continuous abstraction for each location and
that of each transition is expressed by means of polynomials degree
bounded by some fixed bound $d$. The case where the bound is $d=1$
specifies an affine hybrid abstraction $\Xi$. We translate this into a
$d-$closure condition for vector spaces.  Consider a collection of
vector spaces $V_{\ell}: \vspan(B_{\ell})$ for each location $\ell \in
L$.

\begin{definition}
We say that the collection $V_{\ell},
\ell \in L$ is $d-$closed for hybrid system $\Pi$ if and only if
\begin{enumerate}
\item For each location $\ell \in L$, the corresponding vector space
  $V_{\ell}$ is $d$-closed w.r.t to the vector field $\F_{\ell}$
for the continuous subsystem $\scrS_{\ell}$.
 
\item For each transition $t_j: \tupleof{\ell_{pre}, \ell_{post}, G_j,
    F_j}$ and for each element $p \in V_{post}$, we have $
  \fpre(p,t_j) \in \pp{V_{pre}}{d}$.

\end{enumerate}

\end{definition}

Once again, the approach for finding a $d$-closed collection
$V_{\ell},\ \ell \in L$ starts from an initial basis $V_{\ell}^{(0)}$
at each location $\ell$ and refines the basis. Two types of
refinements are applied (a) refinement of $V_{\ell}$ to enforce
closure w.r.t the Lie derivative of its basis elements for the vector
field $\F_{\ell}$ and (b) refinement of $V_{m}$ w.r.t a transition $t:
\tupleof{\ell,m,G,F}$ incoming at location $m$.

\section{Implementation and Evaluation}\label{Section:implementation}

We have implemented the ideas described in this paper to derive affine
abstractions for (a) continuous systems described by ODEs with
polynomial right-hand sides, (b) discrete systems with assignments
that have polynomial RHS and (b) hybrid systems with polynomial ODEs
and discrete transition updates. Our approach takes as inputs the
system description, a degree limit $k > 0$ that is used to construct
the initial basis. Starting from this initial basis, our approach
iteratively applies refinement until convergence. Upon convergence, we
print the basis inferred along with the resulting abstraction. 

Currently, our implementation does not abstract the guard sets of the
transitions and the invariant sets of the ODEs. However, once the
basis is inferred, the abstractions for the guards of the transition
and mode invariants are obtained using quantifier elimination
techniques (which is quite expensive in
practice)~\cite{Collins/75/Quantifier, Collins+Hong/91/Partial,
Dolzmann+Sturm/97/REDLOG} or optimization techniques such as Linear
programming or SOS programming~\cite{Parillo/2003/Semidefinite}. Our
implementation currently relies on manual translation of invariant and
guard assertions into the new basis to form the abstract transition
system.

If a non-trivial abstraction is discovered by our iterative scheme, we
may use a linear invariant generator on the resulting affine system to
infer invariants that relate to the original transition system.

Our implementation and the benchmarks used in the evaluation presented
in this section may be obtained upon request.

\subsection{Continuous Systems}

We first describe experimental results obtained for continuous systems
described by ODEs. Figure~\ref{Table:continuous-expt-results}
summarizes the results on continuous system benchmarks. We collected
nearly $15$ benchmark systems and ran our implementation to
search for a linearizing CoB transformation. We report on the degree
of the monomials in the initial basis, time taken to converge and the
number of polynomials in the final basis that form the transformation
to the abstract system. 

\paragraph{Trivial Transformations Found:} Some of the benchmarks
attempted resulted in trivial final transformations. Examples include
the well-known Fitzhugh-Nagumo neuron model, the vanderpol oscillators
and similar small but complex systems that are known to be
non-integrable.

We now highlight some of the interesting results, while summarizing
all benchmarks in Table~\ref{Table:continuous-expt-results}.

\paragraph{Toda Lattice with Boundary Particles:} The Toda lattice
models an infinite array of point particles such that the position and
velocity of the $n^{th}$ particle are affected by its neighbors the
$(n-1)^{th}$ and $(n+1)^{th}$ particle for $n \in
\mathbb{Z}$~\footnote{See description by G{\"o}ktas and
  Hereman~\cite{Goktas+Hereman/2011/Symbolic} and references therein.}.
We consider a finite version of this lattice with $2$ fixed boundary
particles that are constrained to have a fixed position and zero
velocity and $K$ particles in the middle.  The dynamics for $K=2$
non-fixed particles are given by position variables $y_1,y_2$,
velocities $v_1,v_2$ and extra state variables $u_1,u_2$ to model the
interaction with neighbors.
\[ \begin{array}{rcl rcl rcl }
\frac{dx_1}{dt} &=& v_1 & \frac{d v_1}{dt} &=& v_1 (u_1 - u_2) & \frac{du_1}{dt} &=& - v_1 \\
\frac{dx_2}{dt} &=& v_2 & \frac{d v_2}{dt} &=& v_2 u_2 & \frac{du_2}{dt} &=&  v_1 - v_2 \\
\end{array} \]
In addition, we add time $t$ as a variable to the model with dynamics
$\frac{dt}{dt} = 1$.  Our approach initialized with polynomials of
degree $2$ discovers a basis with $10$ polynomials:
\[ \begin{array}{l}
w_1:\ -2 v_2 - 2 v_1 - u_2^2 + 2 x_1 u_1 + x_2^2,\\
 w_2:\ -2 v_2 - 2 v_1 - u_2^2 + u_1 u_2 + x_2 u_1 + x_1 u_2 + x_1 u_1 + x_1 x_2 \\
w_3:\ -2 v_2 - 2 v_1 + 2 u_1u_2 + 2 x_2 u_2 + 2 x_2 u_1 + x_2^2,\ w_4:\ u_1 + x_1,\\
 w_5:\ 2 v_2 + 2 v_1 + u_1^2 + u_2^2,\ w_6:\ u_2 + x_2 - x_1,\ w_7:\ t,\\
 w_8:\ u_1 t + x_1 t,\ w_9:\ u_2 t + u_1 t + x_2 t,\ w_{10}:\ t^2
\end{array} \]

The resulting abstract system has linear dynamics given by: 
\[ \frac{dw_j}{dt} = 0,\ 1 \leq j \leq 6,\ \frac{d w_7}{dt} = 1,\ \frac{dw_8}{dt} = w_4,\ \frac{dw_9}{dt} = w_4 + w_6,\ \frac{dw_{10}}{dt} = 2 w_7 \,.\]
Results for larger instances are reported in
Table~\ref{Table:continuous-expt-results}.

\paragraph{Quadratic Fermi-Pasta-Ulam-Tsingou System:} Consider a system considered by 
Fermi et al.~\cite{Fermi+Pasta+Ulam/1955/Studies}. The system consists
of a chain of particles at positions $x_1,\ldots,x_N$ with fixed boundary
particles $x_0 =0$ and $x_{N+1} = N+1$.
The dynamics are given by 
\[ \frac{d^2 x_i}{dt^2} = ( x_{i+1} + x_{i-1}  - 2 x_i) + \alpha ( ( x_{i+1}^2 - x_i^2) - (x_i - x_{i-1})^2) \,,\ 1 \leq i \leq N\]
We consider an instantiation with $N=3$, searching for CoB transformations with an initial basis of monomials of degree up to $4$. We obtain a transformation representing a conserved quantity
\[ \begin{array}{c}
\frac{1}{2} ( v_1^2 + v_2^2 + v_3^2) + x_1^2 + x_2^2 + x_3^2  - 3 x_3 ( 1  + 3a - a x_3 ) \\
-x_2x_3 ( 1 +  a x_3 - a x_2)  -x_1x_2( 1 + a x_2 - a x_1) 
\end{array}\,.\]
The abstract system is given by  $ \frac{dw_1}{dt} = 0$.

\paragraph{Two Mass Spring System:} Consider the dynamics of two masses connected 
by a spring to each other and to two fixed walls. The state variables are 
$(x_1,x_2, v_1, v_2)$ indicating the position and velocity of the masses while 
the spring constant $k$ is a parameter. The dynamics are given by
\[ \begin{array}{rcl rcl}
\frac{dx_1}{dt} &=& v_1 & \frac{dx_2}{dt} &=& v_2 \\
\frac{dv_1}{dt} &=& k x_2 - 2 k x_1 & \frac{dv_2}{dt} &=& k ( x_1 - x_2 ) \\
\end{array}\]
Our procedure yields a change of basis transformation 
\[ w_1:\ v_2^2 + v_1^2 + k x_2^2  - 2k  x_1 x_2  + 2k x_1^2,\ w_2:\ v_1 v_2 - \frac{1}{2} v_1^2 - \frac{1}{2} k x_2^2 + 2 k x_1 x_2 - \frac{3}{2} k x_1^2 \]
Both $w_1,w_2$ represent conserved quantities, yielding the abstraction
\[ \frac{dw_1}{dt} = \frac{dw_2}{dt} = 0 \,.\]

\paragraph{Biochemical reaction network:} We consider a
 biochemical reaction network benchmark from Dang et
 al.~\cite{Dang+Maler+Testylier/2010/Accurate}. The ODE along with the
 values are parameters in our model coincide with those used by Dang
 et al. The ODE consists of $12$ variables and roughly $14$
 parameters.  Our search for degree bound $\leq 3$ discovers a
 transformation generated by five basis functions (in roughly $3$
 seconds). 

\paragraph{Collision Avoidance}
We consider the algebraic abstraction of the roundabout mode of a
collision avoidance system analyzed recently by Platzer et
al.~\cite{Platzer+Clarke/09/Differential} and earlier by Tomlin et
al.~\cite{Tomlin+Pappas+Sastry/1998/Conflict}.  The two airplane
collision avoidance system consists of the variables $(x_1,x_2)$
denoting the position of the first aircraft, $(y_1,y_2)$ for the
second aircraft, $(d_1,d_2)$ representing the velocity vector for
aircraft 1 and $(e_1,e_2)$ for aircraft $2$. $\omega,\theta$ abstract
the trigonometric terms. In addition, the parameters $a,b,r_1,r_2$ are
also represented as system variables. The dynamics are modeled by the
following differential equations:
\[\begin{array}{ccccccccc}
x_1'  = d_1 & x_2'= d_2 & d_1' = -\omega d_2 & d_2' = \omega d_1 \\
y_1'  = e_1 & y_2'= e_2 & e_1' = -\theta e_2 & e_2' = \theta e_1 \\
a'=0 & b'=0 & r_1' = 0 & r_2' = 0\\
\end{array}\]

A search for transformations of degree $2$ yields a closed vector
space with 27 basis functions within $0.2$ seconds.  The basis
functions include $a,b,r_1,r_2$ and all degree two terms involving
these. Removing these from the basis, gives us $14$ basis functions that
yield a transformation to a $14$ dimensional affine ODE.

\begin{figure}[t]
{\small
\begin{center}
\begin{tabular}{||l|l|l|| l | l | l || l | l | l || }
\hline
ID & \#V & Deg. & \#B0 & Time & \#B* &  \#B0 & Time & \#B* \\
\hline
 Brusselator & 2 & 3 & 3 & 0.01 & 0 & 25 & 2.8 & 0 \\
 Fitz-Nagumo  & 2 & 3 & 3 & 0.01 & 0 & 25 & 2.6 & 0 \\
Vanderpol & 2 & 3 & 3 & 0.01 & 0 & 25 & 1.9 & 0 \\
Proj-drag & 4 & 2 & 3 & 0.02 & 8 & 10 & 9.7 & 64 \\
Circular & 4 & 2 & 3 & 6 & 0.01 & 10 & 10.6 & 83 \\
Hamiltonian & 5 + 1 & 2 & 3 & 0.02 & 5$\dagger$ & 5 & 1.3 & 20$\dagger$  \\
Two-spring & 4 + 1 & 2 &3 &  0.03 & 2$\dagger$ & 5 & 0.5 & 6$\dagger$ \\
Toda-2 & 7 & 2 & 3 & 0.1 & 22 & 5 & 4.6 & 82 \\
Toda-3 & 10 & 2 & 3 &0.5  & 38 &  5 & 95 & 169 \\
Toda-5 & 16 & 2 & 3 & 6 & 90 & 5 & 6373 & 559 \\
Toda-10 &31 & 2 & 3 &  301.5 & 375 & 5 & \multicolumn{2}{c||}{dnf}    \\
FPUT-3 & 6 + 1& 3 & 3 & 0.05 & 0$\dagger$ & 5 & 3.7 & 2 $\dagger$ \\
FPUT-5 & 10 + 1 & 3 & 3 & 0.4 & 0 & 5 & 231 & 2 \\
Bio-network & 13 & 2 & 3 & 0.07 & 5 & 5 &  4800 & 20 \\
Roundabout & 10 + 4 & 2 & 3 & 1.5 & 68$\dagger$ & 5 & 890 & $\geq 600\ \dagger$ \\
\hline
\end{tabular}
\end{center}
}
\caption{Experimental evaluation results on non linear polynomial ODE
  benchmarks at a glance. Legend: \textbf{\#V} denotes number of
  system variables + parameters, \textbf{Deg.}: max. degree of the
  RHS, \textbf{\#B0:} degree limit for monomials in the initial basis,
  \textbf{Time}: timing in seconds, \textbf{\#B*:} number of elements
  in the final basis, $\dagger$: some elements of the basis involving
  just the parameters were discarded from the count and \textbf{dnf:}
  did not finish in 2hrs or out of memory crash. }\label{Table:continuous-expt-results}
\end{figure}

\subsection{Discrete Systems}

We now describe experimental results on some discrete programs. We
used a set of benchmark programs that require non-linear invariants to
prove correctness compiled by Enric Carbonell~\footnote{The benchmark
  instances are available on-line
  at~\url{http://www.lsi.upc.edu/~erodri/webpage/polynomial_invariants/list.html}.}. Our evaluation focuses on a subset of benchmarks that have non-linear
assignments or guards in them. The methods presented
here converge in a single step with the initial basis whenever the
program being considered already has affine updates.

\begin{figure}[t]
\begin{center}
\begin{tabular}{cc}
{\begin{minipage}{4.5cm}
\small
\begin{lstlisting}
int fermat(int N, int R)
   pre (N >= 0 && R >= 0);
   int u,v,r;
   u := 2*R -1;
   v := 1;
   r := R*R -N;
1: while ( r != 0 ){
2:  while (r > 0)
     (r,v) := (r-v, v+2);
3:  while (r < 0)
     (r,v) := (r+u, u+2);
   }
end
\end{lstlisting}
\end{minipage}} &
\begin{minipage}{5cm}
{\small
\[  \begin{array}{l}
-4  r -  v^2 -4  N v + 2  v +  u^2 -2  u \ \leq\  0 \ \land\\
   -  r -  N u \ \leq\  0 \ \land\  1 -  v^2 \ \leq\  0 \ \land\\
   1 -  uv \ \leq\  0 \ \land\   -  R v +  R \ \leq\  0 \ \land\\
   1 -  v \ \leq\  0 \ \land\   1 -  u^2 \ \leq\  0 \ \land\\
   -  R u + 2  R^2 +  R \ \leq\  0 \ \land\ -  N u \ \leq\  0 \ \land\\
   1 -  u \ \leq\  0 \ \land\   -  R^2 \ \leq\  0 \ \land\   -  N R  \ \leq\  0 \ \land\\
   -  R \ \leq\  0 \ \land\   -  N^2 \ \leq\  0 \ \land\\
   v^2 -2  v -  u^2 + 2  u \ \leq\  0 \ \land\\
   1 +  r -  u -  R^2 \ \leq\  0 \ \land\
   1 + 4  r -  u^2 \ \leq\  0 \ \land\\
   4  r +  v^2 -2  v -  u^2 + 2  u \ \leq\  0 \ \land\\
   2 + 6  r -  uv -  u^2 -2  R^2 \ \leq\  0 \ \land\\
   4  r + v^2 -2  v -  u^2 + 2  u + 4  N= 0 
\end{array} \]
}
\end{minipage}
\end{tabular}
\end{center}
\caption{Fermat's algorithm for prime factorization taken from
  Bressoud~\cite{Bressoud/1989/Factoring} and invariants computed at
  location 1 using polyhedral analysis of the
  linearization.}\label{Fig:fermat-factorization}
\end{figure}
\paragraph{Fermat Factorization:}
Figure~\ref{Fig:fermat-factorization} shows a program for finding a
factor of a number $N$ near its square root taken from a book by
Bressoud~\cite{Bressoud/1989/Factoring}. Our analysis initialized with
monomials of degree up to $2$ over the program variables yields a final
basis consisting of $17$ polynomials. The resulting affine system is
analyzed by a polyhedral analyzer using abstract interpretation to
yield invariants. Some of the invariants obtained at the loop head are
shown in Figure~\ref{Fig:fermat-factorization}. The equality invariant
\[ 4  r + v^2 -2  v -  u^2 + 2  u + 4  N= 0  \]
is obtained at locations $1,2$ and $3$ in the program. This forms a key
part of the program's partial correctness proof. 

\begin{figure}[t]
\begin{center}
\begin{tabular}{cc}
{\begin{minipage}{5cm}
\scriptsize
\begin{lstlisting}
int productBR(int x, int y)
   pre (x >= 0 && y >= 0);
   int a,b,p,q;
   (a,b,p,q) := (x,y,1,0);
1: while ( a >= 1 && b >= 1 ){
     if ( a mod 2 == 0 && b mod 2 == 0)
       (a,b,p) := (a/2, b/2, 4 * p);
     elsif (a mod 2 == 1 && b mod 2 == 0)
        (a,q) := (a-1, q+ b*p);
     elsif (a mod 2 == 0 && b mod 2 == 1)
        (b,q) := (b-1, q + a*p);
     else 
        (a,b,p) := (a-1, b-1, 
                    q + (a+b-1)*p);
   }
end
\end{lstlisting}
\end{minipage}} &
\begin{minipage}{5cm}
{\small
\[  \begin{array}{l}
1 -p^2 \leq 0 \ \land\   -y p + y \leq 0 \ \land\\
   -x p + x \leq 0 \ \land\   1 -p \leq 0 \ \land\\
   -1 -b  \leq 0 \ \land\   -1 -a \leq 0 \ \land\\
   -y \leq 0 \ \land\   -x \leq 0 \ \land\\
   2  a p -x p -2  a + x \leq 0 \ \land\\
  7  a p -x p -7  a -14  x \leq 0 \ \land\\
   7  a p -x p + 8  a -14  x \leq 0 \ \land\\
  16  a p -3  x p -16  a -12  x \leq 0 \ \land\\
   2  b p-y p -2  b  + y \leq 0 \ \land\\
   7  b p-y p -7  b  -14  y \leq 0 \ \land\\
  7  b p-y p + 8  b  -14  y \leq 0 \ \land\\
   16  b p -3  y p -16  b  -12  y \leq 0
\end{array} \]
}
\end{minipage}
\end{tabular}
\end{center}
\caption{A multiplication algorithm and loop invariant computed using polyhedral analysis on the linearization.}\label{Fig:prod-br-example}
\end{figure}
\paragraph{Product of Numbers:} Consider the benchmark shown in
Figure~\ref{Fig:prod-br-example} that seeks to compute the product of
its arguments $x,y$. Our approach initialized using degree $2$
monomials computes an abstract system with $20$ basis polynomials that
in turn yields an affine transition system with $20$
variables. Figure~\ref{Fig:prod-br-example} shows the invariants
computed using polyhedral abstract interpretation. The invariant $q -
abp = 0$ cannot be established by our technique with degree $2$
monomials. On the other hand, it can be established by considering
degree $3$ monomials in the initial basis. The resulting system
however has $60$ variables, making polyhedral analysis of the system
as a whole hard.

\begin{figure}[t]
\begin{center}
\begin{tabular}{cc}
\begin{minipage}{4cm}
{\small 
\begin{lstlisting}
int geoSum(int a, int r, 
                  int n )
   int s := 0;
   int p := a;
   int k := 0;
   while (k < n) 
      s := s + p;
      p := p * r;
      k := k + 1;
   end
\end{lstlisting}
}
\end{minipage}
& \begin{minipage}{5cm} {\small 
\[\begin{array}{l}
 -k^2 + k \ \leq\  0 \ \land\   -2 -k^2 + 3  k \ \leq\  0 \ \land\\
   -6 -k^2 + 5  k \ \leq\  0 \ \land\   -9 -k^2 + 6  k \ \leq\  0 \ \land\\
   -k \ \leq\  0 \ \land\   -r^2 \ \leq\  0 \ \land\   -n \ \leq\  0 \ \land\\
   -1 + k -n \ \leq\  0 \ \land\   -p + s - r s + a = 0 
\end{array}\]}
\end{minipage}
\end{tabular}
\end{center}
\caption{Geometric summation program and computed loop invariant.}\label{Fig:geometric-example}
\end{figure}
\paragraph{Geometric Summation: } Consider the geometric summation
program in Figure~\ref{Fig:geometric-example}. Our approach computes a
linearization with $5$ variables in the abstract system. Polyhedral
analysis of the resulting program yields the invariant $ (1 - r) s = a
- p $. This invariant together with the invariant $ p = a r^k$ (which
cannot be obtained through algebraic reasoning) suffices to prove the
partial correctness of the program.

\begin{figure}[t]
\begin{center}
\begin{tabular}{||l|l|l|l||l|l|l||l|l||}
\hline
\multicolumn{4}{||c||}{System} & \multicolumn{3}{c||}{Linearization}  & \multicolumn{2}{c||}{Analysis} \\
\hline
ID & \#V& \#Trs & Deg & B0 & \#B* & Time & Time & \#I \\
\hline
\hline
Petter2 & 2 & 1 & 2 & 2 & 3 & 0.02 &$\leq$ 0.01  & 10 \\
Petter3 & 2 & 1 & 2 & 3 & 1 & 0.02 & $\leq$ 0.01 & 2\\
Petter3 & 2 & 1 & 3 & 3 & 4 & 0.02 &  $\leq$ 0.01 & 25\\
Geo & 6 & 2 & 2 & 2 & 6 & 0.02 & $\leq$ 0.01 & 9 \\
Fermat & 5 & 6 & 2 & 2 & 17 & 0.04 & 0.5 & 26 \\
Prodbr & 7 & 5 & 2 & 2 & 20 & 0.06 & 2.0 & 19 \\
Euclidex1 & 11 & 5 & 2 & 2 & 51 & 0.66 & \multicolumn{2}{c||}{DNF} \\
\hline
\end{tabular}
\end{center}
\caption{Timings for computing abstractions of discrete systems and
  analyzing the resulting abstractions. Legend: \textbf{\#V} denotes
  number of system variables, \textbf{\#Trs}: number of transitions,
  \textbf{Deg.}: max. degree of the RHS, \textbf{B0:} degree limit for
  monomials in the initial basis, \textbf{Time}: timing in seconds,
  \textbf{\#B*:} number of elements in the final basis, \textbf{\#I:}
  invariants computed and \textbf{DNF:} did not
  finish in 2hrs or out of memory crash. }\label{Fig:expt-results-discrete}
\end{figure}

\section{Conclusion and Future Directions}
Thus far, we have presented an approach that uses Change-Of-Bases
transformation for inferring abstractions of continuous, discrete and
hybrid systems. We have explored the theoretical underpinnings of our
approach, its connections to various invariant generation techniques
presented earlier.  Our previous work presents an extension of the
approach presented in this paper to infer differential inequality
abstractions~\cite{Sankaranarayanan/2011/Automatic}. Similar
extensions for discrete systems remain unexplored. Furthermore, the
use of the abstractions presented here to establish termination for
transition systems is also a promising line of future research.
Future research will also focus on the use of Lie symmetries to reduce
the size of the ansatz or templates used in the search for conserved
quantities and CoB
transformations~\cite{Goktas+Hereman/2011/Symbolic}.

\bibliographystyle{abbrvnat}

\begin{thebibliography}{41}
\providecommand{\natexlab}[1]{#1}
\providecommand{\url}[1]{\texttt{#1}}
\expandafter\ifx\csname urlstyle\endcsname\relax
  \providecommand{\doi}[1]{doi: #1}\else
  \providecommand{\doi}{doi: \begingroup \urlstyle{rm}\Url}\fi

\bibitem[Alur et~al.(2000)Alur, Henzinger, Lafferriere, and
  Pappas]{Alur+Others/2000/Discrete}
R.~Alur, T.~A. Henzinger, G.~Lafferriere, and G.~Pappas.
\newblock Discrete abstractions of hybrid systems.
\newblock \emph{Proc. of {IEEE}}, 88\penalty0 (7):\penalty0 971--984, 2000.

\bibitem[Alur et~al.(2003)Alur, Dang, and
  Ivan{\v{c}}i{\'c}]{Alur+Dang+Ivancic/2003/Counter}
R.~Alur, T.~Dang, and F.~Ivan{\v{c}}i{\'c}.
\newblock Counter-example guided predicate abstraction of hybrid systems.
\newblock In \emph{TACAS}, volume 2619 of \emph{LNCS}, pages 208--223.
  Springer, 2003.

\bibitem[Bagnara et~al.(2005)Bagnara, Rodr{\'\i}guez-Carbonell, and
  Zaffanella]{Bagnara+Carbonell+Zaffanella/2005/Generation}
R.~Bagnara, E.~Rodr{\'\i}guez-Carbonell, and E.~Zaffanella.
\newblock Generation of basic semi-algebraic invariants using convex polyhedra.
\newblock In \emph{12th International Symposium on Static Analysis (SAS'05)},
  volume 3672 of \emph{Lecture Notes in Computer Science}, pages 19--34.
  Springer-Verlag, Sept. 2005.

\bibitem[Berman et~al.(2007)Berman, Halasz, and
  Kumar]{Berman+Halasz+Kumar/2007/MARCO}
S.~Berman, A.~Halasz, and V.~Kumar.
\newblock {MARCO:} a reachability algorithm for multi-affine systems with
  applications to biological systems.
\newblock In \emph{Hybrid Systems: Computation and Control}, volume 4416, pages
  76--89. Springer--Verlag, 2007.

\bibitem[Bressoud(1989)]{Bressoud/1989/Factoring}
D.~M. Bressoud.
\newblock \emph{Factoring and Primality Testing}.
\newblock Springer-Verlag (Undergraduate Texts in Mathematics), 1989.

\bibitem[Clarke et~al.(2003)Clarke, Grumberg, Jha, Lu, and
  Veith]{Clarke+Others/2003/Counterexample}
E.~M. Clarke, O.~Grumberg, S.~Jha, Y.~Lu, and H.~Veith.
\newblock Counterexample-guided abstraction refinement for symbolic model
  checking.
\newblock \emph{J. ACM}, 50\penalty0 (5):\penalty0 752--794, 2003.

\bibitem[Collins(1975)]{Collins/75/Quantifier}
G.~Collins.
\newblock Quantifier elimination for real closed fields by cylindrical
  algebraic decomposition.
\newblock In H.Brakhage, editor, \emph{Automata Theory and Formal Languages},
  volume~33 of \emph{LNCS}, pages 134--183. Springer, 1975.

\bibitem[Collins and Hong(1991)]{Collins+Hong/91/Partial}
G.~E. Collins and H.~Hong.
\newblock Partial cylindrical algebraic decomposition for quantifier
  elimination.
\newblock \emph{Journal of Symbolic Computation}, 12\penalty0 (3):\penalty0
  299--328, sep 1991.

\bibitem[Col\'on(2004)]{Colon/04/Approximating}
M.~Col\'on.
\newblock Approximating the algebraic relational semantics of imperative
  programs.
\newblock In \emph{$11^{th}$ Static Analysis Symposium (SAS'2004)}, volume 3148
  of \emph{LNCS}. Springer, 2004.

\bibitem[Cousot and Halbwachs(1978)]{Cousot+Halbwachs/78/Automatic}
P.~Cousot and N.~Halbwachs.
\newblock Automatic discovery of linear restraints among the variables of a
  program.
\newblock In \emph{POPL'78}, pages 84--97, Jan. 1978.

\bibitem[Cox et~al.(1991)Cox, Little, and O'Shea]{Cox+Others/1991/Ideals}
D.~Cox, J.~Little, and D.~O'Shea.
\newblock \emph{Ideals, Varieties and Algorithms: An Introduction to
  Computational Algebraic Geometry and Commutative Algebra}.
\newblock Springer-Verlag, 1991.

\bibitem[Dang et~al.(2010)Dang, Maler, and
  Testylier]{Dang+Maler+Testylier/2010/Accurate}
T.~Dang, O.~Maler, and R.~Testylier.
\newblock Accurate hybridization of nonlinear systems.
\newblock In \emph{HSCC '10}, pages 11--20. ACM, 2010.

\bibitem[Dolzmann and Sturm(1997)]{Dolzmann+Sturm/97/REDLOG}
A.~Dolzmann and T.~Sturm.
\newblock {REDLOG}: Computer algebra meets computer logic.
\newblock \emph{ACM SIGSAM Bulletin}, 31\penalty0 (2):\penalty0 2--9, June
  1997.

\bibitem[Fermi et~al.(1955)Fermi, Pasta, and
  Ulam]{Fermi+Pasta+Ulam/1955/Studies}
E.~Fermi, J.~Pasta, and S.~Ulam.
\newblock Studies of non-linear problems.
\newblock Document LA-140, Los Alamos National Laboratories, 1955.

\bibitem[G{\"o}ktas and Hereman(2011)]{Goktas+Hereman/2011/Symbolic}
{\"U}.~G{\"o}ktas and W.~A. Hereman.
\newblock Symbolic computation of conservation laws, generalized symmetries,
  and recursion operators for nonlinear differential-difference equations.
\newblock In \emph{Dynamical Systems and Methods}. Springer--Verlag, 2011.

\bibitem[Halbwachs et~al.(1997)Halbwachs, Proy, and
  Roumanoff]{Halbwachs+Proy/97/Verification}
N.~Halbwachs, Y.-E. Proy, and P.~Roumanoff.
\newblock Verification of real-time systems using linear relation analysis.
\newblock \emph{Formal Methods in System Design}, 11\penalty0 (2):\penalty0
  157--185, 1997.

\bibitem[Halmos(1974)]{Halmos/74/Finite}
P.~Halmos.
\newblock \emph{Finite-Dimensional Vector Spaces}.
\newblock Springer-Verlag, 1974.

\bibitem[Kov{\'a}cs(2008)]{Kovacs/2008/Reasoning}
L.~Kov{\'a}cs.
\newblock Reasoning algebraically about p-solvable loops.
\newblock In \emph{Tools and Algorithms for the Construction and Analysis of
  Systems ({TACAS})}, volume 4963 of \emph{Lecture Notes in Computer Science},
  pages 249--264. Springer, 2008.

\bibitem[Kowalski and Steeb(1991)]{Kowalski+Steeb/1991/Non-Linear}
K.~Kowalski and W.-H. Steeb.
\newblock \emph{Non-Linear Dynamical Systems and Carleman Linearization}.
\newblock World Scientific, 1991.

\bibitem[Lee(2003)]{Lee/2003/Introduction}
J.~M. Lee.
\newblock \emph{Introduction to Smooth Manifolds}.
\newblock Graduate Texts in Mathematics. Springer--Verlag, 2003.

\bibitem[Manna and Pnueli(1995)]{Manna+Pnueli/95/Temporal}
Z.~Manna and A.~Pnueli.
\newblock \emph{Temporal Verification of Reactive Systems: {S}afety}.
\newblock Springer, New York, 1995.

\bibitem[Matringe et~al.(2009)Matringe, Moura, and
  Rebiha]{Matringe+Others/2009/Morphisms}
N.~Matringe, A.~V. Moura, and R.~Rebiha.
\newblock Morphisms for non-trivial non-linear invariant generation for
  algebraic hybrid systems.
\newblock In \emph{HSCC}, volume 5469 of \emph{LNCS}, pages 445--449, 2009.

\bibitem[Meiss(2007)]{Meiss/2007/Differential}
J.~D. Meiss.
\newblock \emph{Differential Dynamical Systems}.
\newblock SIAM publishers, 2007.

\bibitem[M{\"u}ller-Olm and Seidl(2004)]{Muller-Olm+Seidl/2004/Precise}
M.~M{\"u}ller-Olm and H.~Seidl.
\newblock Precise interprocedural analysis through linear algebra.
\newblock In \emph{Principles Of Programming Languages ({POPL})}, pages
  330--341. ACM, 2004.
\newblock ISBN 1-58113-729-X.

\bibitem[Oishi et~al.(2008)Oishi, Mitchell, Bayen, and
  Tomlin]{Oishi+Others/2008/Invariance}
M.~Oishi, I.~Mitchell, A.~M. Bayen, and C.~J. Tomlin.
\newblock Invariance-preserving abstractions of hybrid systems: Application to
  user interface design.
\newblock \emph{IEEE Trans. on Control Systems Technology}, 16\penalty0 (2),
  Mar 2008.

\bibitem[Parillo(2003)]{Parillo/2003/Semidefinite}
P.~A. Parillo.
\newblock Semidefinite programming relaxation for semialgebraic problems.
\newblock \emph{Mathematical Programming Ser. B}, 96\penalty0 (2):\penalty0
  293--320, 2003.

\bibitem[Petkovsek et~al.(1996)Petkovsek, Wilf, and
  Zeilberger]{Petkovsek+Wilf+Zeilberger/1996/AB}
M.~Petkovsek, H.~S. Wilf, and D.~Zeilberger.
\newblock \emph{{A=B}}.
\newblock A.K. Peters/ CRC Press, 1996.

\bibitem[Petter(2004)]{Petter/2004/Invarianten}
M.~Petter.
\newblock \emph{Berechnung von polynomiellen invarianten ({German})}.
\newblock PhD thesis, Technische Universit{\"a}t M{\"u}nchen, October 2004.
\newblock Cf. \url{http://www2.cs.tum.edu/~petter/da/da.pdf} .

\bibitem[Platzer and Clarke(2009)]{Platzer+Clarke/09/Differential}
A.~Platzer and E.~Clarke.
\newblock Computing differential invariants of hybrid systems as fixedpoints.
\newblock \emph{Formal Methods in Systems Design}, 35\penalty0 (1):\penalty0
  98--120, 2009.

\bibitem[Prabhakar et~al.(2012)Prabhakar, Dullerud, and
  Viswanathan]{Prabhakar+Viswanathan/2012/Stability}
P.~Prabhakar, G.~Dullerud, and M.~Viswanathan.
\newblock Pre-orders for reasoning about stability.
\newblock In \emph{Hybrid Systems: Computation and Control (HSCC)}. ACM Press,
  2012.
\newblock to appear (April 2012).

\bibitem[Rodr{\'\i}guez-Carbonell and
  Kapur(2004{\natexlab{a}})]{Carbonell+Kapur/2004/Abstract}
E.~Rodr{\'\i}guez-Carbonell and D.~Kapur.
\newblock {An Abstract Interpretation Approach for Automatic Generation of
  Polynomial Invariants}.
\newblock In \emph{{International Symposium on Static Analysis (SAS 2004)}},
  volume 3148 of \emph{Lecture Notes in Computer Science}, pages 280--295.
  Springer-Verlag, 2004{\natexlab{a}}.

\bibitem[Rodr{\'\i}guez-Carbonell and
  Kapur(2004{\natexlab{b}})]{Carbonell+Kapur/2004/Automatic}
E.~Rodr{\'\i}guez-Carbonell and D.~Kapur.
\newblock {Automatic Generation of Polynomial Loop Invariants: Algebraic
  Foundations}.
\newblock In \emph{{International Symposium on Symbolic and Algebraic
  Computation 2004 (ISSAC04)}}, pages 266--273. ACM Press, 2004{\natexlab{b}}.

\bibitem[Sankaranarayanan(2010)]{Sankaranarayanan/2010/Invariant}
S.~Sankaranarayanan.
\newblock Automatic invariant generation for hybrid systems using ideal fixed
  points.
\newblock In \emph{Hybrid Systems: Computation and Control}, pages 211--230.
  ACM Press, 2010.

\bibitem[Sankaranarayanan(2011)]{Sankaranarayanan/2011/Automatic}
S.~Sankaranarayanan.
\newblock Automatic abstraction of non-linear systems using change of bases
  transformations.
\newblock In \emph{Hybrid Systems: Computation and Control ({HSCC})}, pages
  143--152. ACM, 2011.

\bibitem[Sankaranarayanan et~al.(2004)Sankaranarayanan, Sipma, and
  Manna]{Sankaranarayanan+Others/2004/Non-Linear}
S.~Sankaranarayanan, H.~Sipma, and Z.~Manna.
\newblock Non-linear loop invariant generation using {G}r{\"o}bner bases.
\newblock In \emph{ACM Principles of Programming Languages (POPL)}, pages
  318--330. ACM Press, 2004.

\bibitem[Sankaranarayanan et~al.(2006{\natexlab{a}})Sankaranarayanan,
  Col{\'o}n, Sipma, and
  Manna]{Sankaranarayanan+Colon+Sipma+Manna/2006/Efficient}
S.~Sankaranarayanan, M.~A. Col{\'o}n, H.~Sipma, and Z.~Manna.
\newblock Efficient strongly relational polyhedral analysis.
\newblock In \emph{VMCAI'06}, volume 3855/2006 of \emph{LNCS}, pages 111--125,
  2006{\natexlab{a}}.

\bibitem[Sankaranarayanan et~al.(2006{\natexlab{b}})Sankaranarayanan, Sipma,
  and Manna]{Sankaranarayanan+Sipma+Manna/06/Fixed}
S.~Sankaranarayanan, H.~B. Sipma, and Z.~Manna.
\newblock Fixed point iteration for computing the time-elapse operator.
\newblock In \emph{HSCC}, LNCS. Springer, 2006{\natexlab{b}}.

\bibitem[Sankaranarayanan et~al.(2008)Sankaranarayanan, Sipma, and
  Manna]{Sankaranarayanan+Sipma+Manna/2008/Constructing}
S.~Sankaranarayanan, H.~Sipma, and Z.~Manna.
\newblock Constructing invariants for hybrid systems.
\newblock \emph{Formal Methods in System Design}, 32\penalty0 (1):\penalty0
  25--55, 2008.

\bibitem[Tiwari(2008)]{Tiwari/2008/Abstractions}
A.~Tiwari.
\newblock Abstractions for hybrid systems.
\newblock \emph{Formal Methods in Systems Design}, 32:\penalty0 57--83, 2008.

\bibitem[Tomlin et~al.(1998)Tomlin, Pappas, and
  Sastry]{Tomlin+Pappas+Sastry/1998/Conflict}
C.~J. Tomlin, G.~J. Pappas, and S.~Sastry.
\newblock Conflict resolution for air traffic management: A study in
  multi-agent hybrid systems.
\newblock \emph{{IEEE} Trans. on Aut. Control}, 43\penalty0 (4):\penalty0
  509--521, April 1998.

\bibitem[Vassilyev and Ul{}'yanov(2009)]{Vassilyev+Ulyanov/2009/Preservation}
S.~Vassilyev and S.~Ul{}'yanov.
\newblock Preservation of stability of dynamical systems under homomorphisms.
\newblock \emph{Differential Equations}, 45:\penalty0 1709--1720, 2009.

\end{thebibliography}

\end{document}